\documentclass[a4paper,11pt]{article}

\usepackage{relsize}
\usepackage{epigraph}
\usepackage{nicefrac}
\usepackage[margin=1in]{geometry}
\usepackage{amsmath,amssymb,amsthm}
\usepackage{hyperref}
\usepackage{fullpage}
\hypersetup{colorlinks=true,linkcolor=black,citecolor=blue}

\usepackage{mathtools}
\usepackage{float}
\usepackage[titletoc,title]{appendix}
\usepackage{graphicx, color}
\usepackage[classfont=sanserif,langfont=roman,funcfont=italic]{complexity}
\usepackage{caption}
\usepackage{subcaption}
\usepackage{setspace}

\floatstyle{ruled}
\newfloat{algorithm}{ht}{loa}
\providecommand{\algorithmname}{Algorithm}
\floatname{algorithm}{\protect\algorithmname}
\newcount\Comments
\newcommand{\kibitz}[2]{\ifnum\Comments=1\textcolor{#1}{#2}\fi}

\usepackage[dvipsnames]{xcolor}
\colorlet{LightRubineRed}{RubineRed!70!}

\usepackage{comment}
\usepackage{enumerate}
\usepackage{amsmath, amsthm, amssymb, amstext,  graphicx}

\newtheorem{theorem}{Theorem}[section]

\newtheorem{definition}{Definition}
\newtheorem{proposition}{Proposition}
\newtheorem{corollary}{Corollary}

\newtheorem{lemma}{Lemma}

\newtheorem{open}{Open Problem}

\usepackage{multirow}
\usepackage{hyperref}
\hypersetup{
	colorlinks   = true, 
	urlcolor     = blue, 
	linkcolor    = blue, 
	citecolor   = blue 
}
\usepackage{xcolor}
\renewcommand{\vec}[1]{\mathbf{#1}}

\def\lc{\left\lceil}   
\def\rc{\right\rceil}

\title{Communication Complexity of Cake Cutting\footnote{This project has received funding from the European Research Council (ERC) under the European Union’s
Horizon 2020 research and innovation programme (grant agreement No 740282). 
The project was also supported by the ISF grant 1435/14 administered by the Israeli Academy of Sciences and Israel-USA Bi-national Science Foundation (BSF) grant 2014389.
Simina also acknowledges support from the I-CORE Program of the Planning and Budgeting Committee and The Israel Science Foundation. }
}

\author{
	Simina Br\^anzei\footnote{Purdue University, USA. E-mail: \textcolor{blue}{\href{mailto:simina.branzei@gmail.com}{simina.branzei@gmail.com}}.}\\
	\newline
	\and
	Noam Nisan\footnote{Hebrew University of Jerusalem and Microsoft Research, Israel. E-mail: 
		\textcolor{blue}{\href{mailto:noam@cs.huji.ac.il}{noam@cs.huji.ac.il}}.}
}

\date{}

\begin{document}

\maketitle

\begin{abstract}
We study classic cake-cutting problems, but in discrete models rather than using infinite-precision real values, specifically, focusing on their communication complexity. Using general discrete simulations of classical infinite-precision protocols (Robertson-Webb and moving-knife), we roughly partition the various
fair-allocation problems into 3 classes: ``easy" (constant number of rounds of logarithmic many bits), ``medium" (poly-logarithmic total communication), and ``hard".
Our main technical result concerns two of the ``medium" problems (perfect allocation for 2 players and equitable allocation for any number of players) which we prove 
are not in the ``easy" class.  Our main open problem is to separate the ``hard" from the ``medium" classes.
\end{abstract}

\section{Introduction}
The concept of dividing valuable resources in a way that is fair has a long history in economic thinking and underpins the very notion of justice in society \cite{Foley67,Rawls71,Varian74,Moulin03}. An instantiation of this challenge is the problem of allocating a set of goods---such as land, time, mineral deposits, or an inheritance---among multiple participants with equal rights but distinct interests. 
The model that encapsulates the pure fair division 
problem\footnote{In economies with production, the participants may have different contributions to the product of society, so their claims may not necessarily be equal, but our focus here will be on equal rights.} is due to Steinhaus \cite{Steinhaus48} under the name of cake cutting. The ``cake'' is a metaphor for the collection of interest, formally defined as the unit interval, while the value of each participant (player) for a piece of cake is given by the integral of a private probability measure over that piece. The cake cutting model can be seen as the limit of a discrete model of fair division---in which the players have additive valuations over a finite set of indivisible goods---by letting the number of goods (i.e. ``atoms'') grow to infinity.

This model has given rise to a rich literature concerned with understanding what fairness itself means and what procedures can materialize it \cite{BT96,RW98,Pro13}. The outcome of a protocol is an allocation, which is an assignment of (non-overlapping) pieces to the players, so that each piece is a union of intervals  and the whole cake is allocated.
Four of the commonly considered notions of fairness are proportionality, envy-freeness, equitability, and perfection.
A simple envy-free protocol is the Cut-and-Choose procedure for two players, Alice and Bob, in which a mediator asks Alice to cut the cake in two pieces of equal value to her, after which Bob is allowed to take his favorite while Alice is left with the remainder.

Moving beyond Cut-and-Choose, protocols for more players can get very complex and, in fact,
a central question in the cake cutting literature is that of computing fair allocations despite the informational challenge of private preferences. The goal of a mediator (or center) is to help the parties reach a fair solution, which in turn requires it to learn enough valuation information.  
The existing cake cutting protocols are broadly classified in three categories as follows.
The largest class contains so called ``discrete'' protocols, which operate in a query model (due to Robertson and Webb \cite{RW98,WS07}), where the center can repeatedly ask the players to make a cut or evaluate an existing piece.
Despite the fact that the Robertson-Webb (RW) query model is viewed as discrete in the cake cutting literature, it is an infinite precision model, which is unavoidable for the purpose of computing exact 
solutions (or even approximate ones for arbitrarily complex instances.)
Continuous (moving knife) procedures generally involve multiple knives sliding over the cake, with stopping conditions that signal the presence of a fair allocation. 
In a companion paper \cite{BN17} we formalize a general class of moving knife procedures 
and prove that they can be approximately simulated by logarithmically many queries in the RW model.
Finally, direct revelation protocols were studied also in the context of mechanism design \cite{MT10,MN12,CLPP13,BM15}, where the players directly submit their entire valuation to the center, whose goal is to compute a fair allocation despite the strategic behavior of the players.
The positive results in the direct revelation model are either only for very restricted preferences, such as piecewise uniform functions \cite{CLPP13}, or rely again on infinite precision for general preferences \cite{MT10}.

In this paper we study the communication complexity \cite{KN96,LS09} of cake cutting: we assume that each player 
knows its own valuation function, and
study the amount of communication that the players need to exchange (between themselves or with a mediator)
in order to find an $\epsilon$-fair allocation, for various notions of fairness.  

There are several motivations for studying the communication complexity of fair allocation problems.
First, it is natural to use discrete protocols for studying notions of complexity rather than
infinite-precision ones like the RW model, and communication complexity is the strogest, most general, model that allows only discrete queries 
to each player's valuation\footnote{Any query that has $k$ possible answers to the valuation of player $i$ can be simulated
in a communication protocol by $\log_2 k$
bits of communication from player $i$.}.
Second,
we will show that communication lower bounds also imply lower bounds not only in the RW model but also in various moving-knife models, allowing
us to utilize the strong tools already developed in the theory of communication complexity for proving results on cake cutting in 
other models. 
Third, cake cutting problems seem to be interesting challenges from the point of view of communication complexity:
these problems have low non-deterministic communication complexity\footnote{I.e. specifying a fair allocation requires a small
number of bits, and its fairness can be easily verified.}, and there are rather few such problems for which good
deterministic or randomized lower bounds are known (specifically \cite{BR17, RW16} and Karchmer-Wigderson games \cite{KW90}.)  
Since the existence proofs of fair allocations are based on fixed point theorems such as Sperner, Brower, or Borsuk-Ulam \cite{Neyman46,simmons80,Stromquist80,ALON1987247,Su99,DQS12,FFGZ16,FG18}, their study is a basic instance of understanding various
 ``PPAD-like'' classes \cite{P94} in the communication
complexity model.
Finally, this work is part of the general thread of understanding the communication complexity of economic problems, such as 
markets, auctions \cite{NS06,BNS07,DV13,DNS14,ANRW15,Assadi17,BMW17}, matchings \cite{GNOW15}, and voting \cite{CS05}.

It is clear that in discrete models such as communication complexity exact fairness is impossible to achieve and
we thus focus on the approximate notion, $\epsilon$-fairness.
We note that understanding $\epsilon$-fairness is motivated by the fact that in reality the goods to be 
allocated are rarely completely continuous but rather, at best, are fine enough so that a continuous model
is a good approximation. Similarly, in reality we do not need (or can really get) exact fairness, but rather approximate fairness
should suffice.  It is thus interesting, from a complexity point of view, to study the complexity of achieving various notions of
fairness as a function of the desired approximation parameter $\epsilon$. While much previous work has studied the complexity of
cake cutting notions in 
terms of the number of players $n$ \cite{BT96,RW98,EP06b,WS07,Pro09,AM16}, recent work \cite{CP12,PW17,BN17} started 
studying the dependence in terms of $\epsilon$ as well.  
In this paper we will always treat $n$ as a small fixed constant, and focus on the complexity in terms of the desired approximation
level $\epsilon$. 
There
does remain work in understanding the joint dependence on both $n$ and $\epsilon$.

\subsection{Our Results}

We study both general notions of fairness and in particular the following four standard notions:  
proportionality, envy-freeness, equitability, and perfection \cite{RW98}:
\begin{definition}
An allocation $A_1 \ldots A_n$ of the cake among players $1 \ldots n$, where each player $i$ receives piece $A_i$, is called:
\begin{itemize}
\item $\epsilon$-proportional if for every $i$: $v_i(A_i) \ge 1/n - \epsilon$.
\item $\epsilon$-envy-free if for every $i,j$: $v_i(A_i) \ge v_i(A_j) - \epsilon$.
\item $\epsilon$-equitable if for every $i,j$: $|v_i(A_i) - v_j(A_j)| \le \epsilon$.\footnote{The equitability requirement is considered particularly relevant in situations of heightened conflict, 
such as divorce settlements or war negotiations.}
\item $\epsilon$-perfect if for every $i,j$: $|v_i(A_j) - 1/n| \le \epsilon$.
\end{itemize}
\end{definition}


Our first contribution is a careful definition of the communication complexity model for cake cutting. 
We observe for the problem to make sense we must have the 
following two features be part of the setting: 
the valuations of the players must have bounded density
(this is the correct equivalent to the no-atoms assumption in the usual cake cutting
model) and that the number of cuts (pieces, intervals) in the produced allocation must be bounded.  

\medskip

Without bounded density, nothing can be done, as is captured by the following proposition.

\begin{proposition}
	Let $\epsilon_0 = 1/3$ be a constant approximation parameter. No finite discrete protocol can output a
	$\epsilon_0$-proportional allocation between $n=2$ players for all atom-less valuations (but whose densities are not necessarily bounded).
	The same impossibility holds for every other number of players $n$, $\epsilon<1/2$, and any other standard fairness criteria.
\end{proposition}

On the other hand, if the protocol is allowed to produce allocations with an unbounded number of
cuts then allocations that are $\epsilon$-fair can be found with high probability without any communication.

\begin{proposition}
	For every $n$ and $\epsilon$, there exists a randomized protocol, with public coins, that without any communication or 
	queries produces with high probability an $\epsilon$-perfect allocation
	with at most $O(\epsilon^{-2})$ cuts. The
	protocol may be turned into one with private coins by adding a single round of communication of $O(\log \epsilon^{-1})$ bits. 
     The same holds for any other standard fairness criteria.
\end{proposition}

From this point on we will therefore assume valuations with bounded densities and protocols that produce allocations
with a bounded number of cuts.

\medskip

Our second set of contributions is providing general simulation results that apply to a rather general class of (abstract)
fairness notions, including all those studied in the literature (the formal requirement for such general fairness notions can be found in Section \ref{sec:model_fairnessnotions}).
We start with a general simulation of the (infinite precision)
RW model by the (discrete, finite precision) communication model. 


\begin{theorem} 
Let $\mathcal{M}$ be any protocol in the RW model that makes at most $r$ queries and outputs $\mathcal{F}$-fair allocations, where $\mathcal{F}$ is a fairness notion. Then there exists a communication protocol $\mathcal{M}_{\epsilon}$ with $r$ rounds of $O(\log \epsilon^{-1})$ bits of communication each, such that $\mathcal{M}_{\epsilon}$ outputs $\epsilon$-$\mathcal{F}$-fair allocations.
\end{theorem}

A more general class of protocols is that of moving knife protocols, in which ``knives''
are continuously moved across the cake until some condition is met (perhaps in addition to RW queries).
In a companion paper \cite{BN17} we defined a general class of such protocols that can be approximately simulated in the RW model,
and here we define an even more general class that can be approximately simulated with small communication. 
All known moving-knife procedures fall into this class, using a constant number of steps.
While the formal definition is a bit cumbersome, in an intuitive sense it captures protocols that rely on the simple intermediate value theorem 
(rather than a ``full-fledged'' fixed point theorem.)

\begin{theorem}
Let $\mathcal{M}$ be any moving knife protocol that runs in constantly many steps and outputs $\mathcal{F}$-fair partitions. Then for every $\epsilon > 0$, there is a communication protocol $\mathcal{M}_{\epsilon}$ that uses 
	$O(\log \epsilon^{-1})$ rounds each with $O(\log \epsilon^{-1})$ bits of communication and 
	outputs $\epsilon$-$\mathcal{F}$-fair partitions. 
\end{theorem}


For fairness notions for which no protocols are known in existing models,
we still show that a fair allocation can be found with nearly linear communication.

\begin{proposition} \label{lem:sim_all}
	For any fairness notion $\mathcal{F}$ (that is guaranteed to always exist and has at most a constant $C$ number of cuts), 
an  $\epsilon$-$\mathcal{F}$-fair allocation can be found using a single round
	of $O(\epsilon^{-1}\log\epsilon^{-1})$ bits of communication.
\end{proposition}

Armed with these general simulation results we get a high-level map of the communication complexity of various notions of fairness. 
We qualitatively categorize the standard fairness criteria 
in terms of their communication complexity into three categories: 
``Easy'', ``Medium'', and ``Hard''. Let us discuss each of these categories.

\vspace{0.1in}
\noindent
{\bf Easy:}
\vspace{0.1in}

The first class of problems consists of fair division problems that have known bounded protocols in the RW model, i.e. can be solved with a constant number of rounds of logarithmic (or even poly-logarithmic) bits.
Applying our simulation theorem to known protocols in the RW model \cite{AM16,EP84} we get:

\smallskip

\begin{theorem}
	The following problems have communication protocols with a constant number of rounds of communication of $O(\log \epsilon^{-1})$
	bits each round:
	\begin{enumerate}
	    \item For any fixed number of players $n$, a connected $\epsilon$-proportional allocation among $n$ players.
	    \item For any fixed number of players $n$, for some constant $C$ that depends on $n$, 
                  an $\epsilon$-envy-free allocation with at most $C$ cuts. 
	\end{enumerate}
\end{theorem}

We show that this upper bound is tight in a strong way:

\begin{theorem}
	Every (deterministic or randomized) protocol for (not-necessarily connected) $\epsilon$-proportional allocation among any $n \ge 2$ 
	players requires $\Omega(\log \epsilon^{-1})$ bits of communication (with no restriction
	on the number of rounds.) 
\end{theorem}

\vspace{0.1in}
\noindent
{\bf Medium:}
\vspace{0.1in}

The second class corresponds to problems that have known moving-knife protocols: computing a connected equitable allocation between $n=2$ players, an perfect allocation with two cuts (the minimum possible) between $n=2$ players, and a connected envy-free allocation among $n=3$ players.
We show these problems have logarithmic or poly-logarithmic total communication complexity:
equitable and envy-free can be solved approximately with $O(\log{\epsilon^{-1}})$ bits,
perfect with $O(\log^2 \epsilon^{-1})$, and general moving knife protocols with $polylog(\epsilon^{-1})$.

Our main technical result is a lower bound that applies to the two player problems above (i.e. perfect and equitable) showing that they are not in the ``easy class'' and thus separating the ``easy'' and ``medium'' classes.

\medskip

Our technical analysis of these problems proceeds by introducing a natural discrete communication problem that we call the ``crossing problem''.

\begin{definition}
	The {\em general} $crossing(m)$ problem is the following: Alice gets a sequence of integer 
      numbers $x_0,x_1, \ldots ,x_m$ with $0 \le x_i \le m$ for all $i$ and similarly Bob gets 
	a sequence of integer numbers $y_0,y_1, \ldots ,y_m$ with $0 \le y_i \le m$ for all $i$, and such 
     that $x_0 \le y_0$ and $x_m \ge y_m$. Their goal is to find an index $1 \le i \le m$ 
	such that either both $x_{i-1} \le y_{i-1}$ and $x_i \ge y_i$ or that both
	$x_{i-1} \ge y_{i-1}$ and $x_i \le y_i$.\footnote{An even more difficult variant of this problem
		would insist that the output is the former rather than the latter.   This variant does not
		seem to correspond directly to any fair division problem so we will not name it.  It is worth
		noting that the protocol in the proof of theorem \ref{ub} solves this harder variant as well.}
	
	In the monotone variant, $mon-crossing(m)$, Alice and Bob are also assured that 
     $0=x_0 \le x_1 \le x_2 \le \cdots \le x_m=m$ and $m=y_0 \ge y_1 \ge \cdots \ge y_m=0$ and again need to find
	an index $1 \le i \le m$ such that $x_{i-1} \le y_{i-1}$ and $x_i \ge y_i$.
\end{definition} 

Our upper bounds on the communication complexity of the crossing problem are as follows.

\begin{theorem} 
	The deterministic communication complexity of $crossing(m)$ is $O(\log^2(m))$.  The randomized
	communication complexity of $crossing(m)$ is $O(\log(m) \log\log(m))$.
\end{theorem}

\begin{theorem}
	The (deterministic) communication complexity of the $mon-crossing(m)$ problem is $O(\log(m))$.
\end{theorem}

In terms of total communication we have a logarithmic lower bound, so closing the gap for the general crossing problem remains an open problem.
Our main technical result is a lower bound that shows that even the monotone variant of the crossing problem cannot be solved using
a constant number of rounds of logarithmic many bits of communication each and thus is not in the ``easy'' class
described above.  Our proof uses ``round elimination'' lemmas from communication complexity \cite{mnsw95,ss08}.

\begin{theorem}
	Any (deterministic or randomized) protocol that solves $mon-crossing(m)$ using rounds of 
	communication of
	$t$-bits each (with $\log m \le t \le m$) requires $\Omega(\log m/ \log t)$ rounds of communication.
\end{theorem}

It turns out the the monotone crossing problem exactly captures the complexity of finding a 
connected equitable allocation between $n=2$ players, while the general crossing problem captures exactly the
complexity of finding a perfect allocation (with 2-cuts) between $n=2$ players. 
The monotone variant also
implies an upper bound for finding a connected envy-free allocation among $n=3$ players. 

\begin{corollary} 
	The deterministic communication complexity of finding an $\epsilon$-perfect allocation with two cuts between $n=2$ players is $O(\log^2 \epsilon^{-1})$. The randomized communication complexity of this problem is $O(\log \epsilon^{-1} \log\log \epsilon^{-1})$.
\end{corollary}

\begin{corollary}
	The deterministic communication complexity of finding a connected $\epsilon$-equitable allocation 
	between $n = 2$ players and a connected $\epsilon$-envy-free allocation 
	among $n = 3$ players is $O(\log \epsilon^{-1})$.
\end{corollary}
	
\begin{corollary}
Any (deterministic or randomized) protocol for finding an $\epsilon$-perfect allocation with 2-cuts between $n=2$ players or for finding a connected $\epsilon$-equitable allocation between $n=2$ players using rounds of  communication of
$polylog(\epsilon^{-1})$-bits each requires $\Omega(\log \epsilon^{-1}/ \log \log \epsilon^{-1})$ rounds of communication.	
\end{corollary}

In particular, this theorem implies nearly logarithmic lower bounds on the number of queries in the RW model 
for these two problems.
Comparing our result for equitability to the recent lower bound of $\Omega(\log \epsilon^{-1}/\log\log \epsilon^{-1})$ by \cite{PW17} 
in the RW model,
our result has two advantages and one disadvantage. The first advantage is that the communication complexity model is 
significantly stronger,
and the second is that our lower bound applies also to randomized protocols (with public coins and two-sided error). The disadvantage
is that our lower bound is only for connected allocations.\footnote{Notice however that as mentioned above
if an arbitrary number of cuts is allowed (a finite number
that may depend on $\epsilon$) then the lower bound is false in a very strong sense.}

We conjecture that the $\epsilon$-equitable and $\epsilon$-perfect allocation problems remain of 
``medium'' complexity (rather than ``easy'') for any fixed constant
number of cuts (that does not depend on $\epsilon$).  
Similarly we conjecture that for $n=3$ players the envy-free problem with connected pieces is not ``easy''.

\vspace{0.1in}
\noindent
{\bf Hard:}
\vspace{0.1in}

The final class of problems are ones where reasonable protocols are not known. Our general result above states that nearly
linear (in $\epsilon^{-1}$) many bits of communication suffice for essentially any fairness criteria.
The main open problem left by this paper is to prove a separation between the ``medium class'' and the ``hard class''.
I.e., to 
prove a super-logarithmic lower bound on the communication 
complexity of some fairness criteria.  This would in particular imply that such a problem
has no moving knife protocol for it (in a rather general class of moving knife procedures), a type of result that is not
known.  Natural candidates for such 
problems are those that have no known protocols for them:

\begin{enumerate}
	\item An $\epsilon$-perfect allocation among $n \ge 3$ players and any (bounded) number of cuts.
	\item A connected $\epsilon$-envy-free allocation among $n \ge 4$ players 
\end{enumerate}

\section{Model and Preliminaries} \label{sec:model}

We have a cake, represented by the interval $[0,1]$, and a set of players $N = \{1, \ldots, n\}$. Each player $i$ has a value density function $v_i : [0,1] \rightarrow \Re_+$. We denote the value of player $i$ for any interval $[a,b] \subseteq [0,1]$ by $v_i([a,b]) = \int_a^b v_i(x) \; dx$. The valuations are additive, so $v_i\left(\bigcup_{j=1}^m X_j\right) = \sum_{j=1}^m v_i(X_j)$ for any disjoint intervals $X_1, \ldots, X_m \subseteq [0,1]$. Note that, by definition, atoms are worth zero, and w.l.o.g. the valuations are normalized so that $v_i([0,1])=1$ for each player $i$.

\smallskip

A \emph{piece} of cake is a finite union of disjoint intervals. A piece is \emph{connected} if it's a single interval.
An \emph{allocation} $A = (A_1, \ldots, A_n)$ is a partition of the cake among the players, such that each player $i$ receives piece $A_i$, the pieces are
disjoint, and $\bigcup_{i \in N} A_i = [0,1]$.\footnote{Some works also consider the problem of fair division where cake can be discarded; see, e.g., \cite{CLPP13,HHA16}.}
An allocation $A$ has $C$ cuts if $C$ is the minimum number of cuts needed to demarcate $A$; e.g. the allocation $A = (A_1, A_2)$ where $A_1 = [0.3,0.5]$ and $A_2 = [0,0.3] \cup [0.5,1]$ has $2$ cuts, at $0.3$ and $0.5$.\footnote{The number of cuts used to demarcate the allocation $A$ would not increase if the piece $A_1 = [0.3,0.4]$ were expressed as $A_1 = [0.3,0.4] \cup [0.4,0.5]$.}

\smallskip

We assume that 
the valuations are private information. We will always have 
\begin{itemize}
\item Bounded density: For some fixed constant $D$ (that may depend on $n$), 
for all $0 \leq a < b \leq 1$, we have that $v_i([a,b]) \le D \cdot (b-a)$.  
\end{itemize}

A valuation $v_i$ is \emph{hungry} if it has positive density everywhere: $v_i(x) > 0$ for all $x \in [0,1]$.

\subsection{Fairness Notions} \label{sec:model_fairnessnotions}

Given player valuations $v_1 \ldots v_n$, our goal is to partition the cake
into pieces $A_1 \ldots A_n$ in a way that satisfies certain fairness criteria, at least approximately.
The following are some of the most commonly considered notions of fairness:

\begin{itemize}
	\item $\epsilon$-proportionality: For all $i$: $v_i(A_i) \ge 1/n - \epsilon$.
	\item $\epsilon$-envy-freeness: For all $i$ and $j$: $v_i(A_i) \ge v_i(A_j) - \epsilon$.
	\item $\epsilon$-equitability: For all $i$ and $j$: $|v_i(A_i) - v_j(A_j)| \le \epsilon$.
	\item $\epsilon$-perfection: For all $i$ and $j$: $|v_i(A_j) - 1/n| \le \epsilon$.
\end{itemize}

For every profile of valuations, it is possible to achieve each one of 
these fairness notions exactly, i.e. with $\epsilon=0$ (see, e.g., \cite{BT96,RW98,Pro13}). The question is how many cuts are needed. In the simplest case, we can give each participant
a connected piece of the cake (i.e. $A_i = [a_i,b_i]$, for all $i$ and some interval $[a_i,b_i]$), but this is not always possible and for some notions of fairness the cake must be cut into a larger number of sub-intervals. Proportional cake cutting, envy-free cake cutting, and equitable cake cutting are each possible with  connected pieces (i.e. $n-1$ cuts), while perfect cake-cutting is possible
with $n(n-1)$ cuts (but not less). Clearly the more cuts do we allow a protocol
to make, the easier its job may become. 

The number of players will be fixed and arbitrary. The number of cuts for any allocation that a protocol can output will be bounded by a constant which may depend on the number of players.
Finally, some of our results will be about general (abstract) notions of fairness, which encompass the four notions above. We will require from a fairness notion that it admits an approximate version as follows.

\begin{definition}[Abstract fairness]
	An abstract fairness notion $\mathcal{F}$ must satisfy the condition:
	\begin{itemize}
		\item Let $A = (A_1, \ldots, A_n)$ be any allocation that is $\mathcal{F}$-fair with respect to some valuations $\vec{v} = (v_1 \ldots v_n)$. Then for any $\epsilon > 0$ and valuations $\vec{v}' = (v_1' \ldots v_n')$, if $|v_i(A_j) - v_i'(A_j)| \leq \epsilon$ $\forall i,j \in N$, the allocation $A$ is also $\epsilon$-$\mathcal{F}$-fair with respect to valuations $\vec{v}'$.
	\end{itemize}
\end{definition}

\subsection{Communication Complexity}

The players will embark on a communication protocol. Each player $i$ starts by knowing its own valuation $v_i$ and the communication task is to find an allocation that is (approximately) fair from the point of view of all the players.
We will count not only bits of communication 
but also the number of rounds of interaction.
Some of our protocols are randomized (with public coins) and need to 
compute the desired output with a specified probability of 
error. 

Fair allocation problems are relations and their communication complexity will be defined next 
(see \cite{KN96}, chapters 3 and 5).   Each party knows its own valuation function. For randomized protocols, each party is also allowed to flip a random coin; that is, each party has access to a private random string of some arbitrary length, where the strings are chosen independently according to some probability distribution. For example, in the case of two parties, Alice and Bob, holding valuations $v_{Alice}$ and $v_{Bob}$, and random strings $r_A$ and $r_B$ of arbitrary lengths, each combination of $v_{Alice}, v_{Bob}, r_A, r_B$ determines a leaf in the protocol tree, which is marked with some allocation $A$. The allocation $A$ is defined as the output of the protocol on input $(v_{Alice}, v_{Bob})$.

\begin{definition}[Deterministic and randomized communication complexity]
	For any notion of fairness $\mathcal{F}$,
	a protocol $\mathcal{P}$ computes an allocation that is fair according to $\mathcal{F}$ if for every profile of valuations $v_1 \ldots v_n$, the protocol reaches a leaf marked by an allocation $A$ that is $\mathcal{F}$-fair with respect to the valuations $v_1 \ldots v_n$. 
	The \emph{deterministic communication complexity} of the fairness notion $\mathcal{F}$, denoted $D(\mathcal{F})$, is the number of bits sent on the worst case input by the best protocol that computes $\mathcal{F}$-fair allocations. 
	
	The \emph{randomized $\epsilon$-error communication complexity} of a fairness notion $\mathcal{F}$, $R_{\epsilon}(\mathcal{F})$, is the worst case number of bits sent by the best randomized protocol that computes $\mathcal{F}$-fair allocations with probability $1-\epsilon$. The error probability is taken over the random choices of the protocol on the worst case input.
\end{definition}

In the definition of randomized protocols above, each party has its own random coins to flip; this is referred to as the ``private'' coin model. It is also possible to allow the parties to have a ``public'' coin, where all of them can see the results of a single series of random coin flips.

\smallskip
We will always require that our protocols produce allocations with at most a bounded number of
cuts $C$, where $C$ may depend on $n$, but not on $\epsilon$.

\begin{definition}
	A communication protocol is a 
	$[t,r]$-protocol if it has $r$ rounds of communication each using at most $t$ bits of communication.  
\end{definition}

\subsection{Why Bound the Density?}
	
	The bounded density requirement is
	the analogue of having no atoms.  The necessity of bounding the density can be demonstrated by the following impossibility result that applies not only to communication protocols but also 
	to any type of protocol that uses any type 
	of queries that have a finite number of possible answers  
	(in contrast to the usual RW protocols in which answers to queries are real  
	numbers). 
	
	\begin{definition} \label{def:reasonable}
		A cake cutting problem is called reasonable if whenever all players have all their weight 
		on some sub-interval $[x,y]$, i.e. $v_1([x,y])=v_2([x,y])= \cdots = v_n([x,y]) = 1$, then the 
		allocation must
		cut the cake somewhere within this interval (and perhaps at finitely many other locations as well).
	\end{definition}
	
	Notice that all studied notions of approximate fairness are reasonable in this sense. For concreteness
	think about the 
	simplest challenge of $(\epsilon=1/3)$-approximate proportional cake 
	cutting between $n=2$ players.
	
	\begin{lemma} \label{lem:density_bound}
		A finite query protocol that always makes only a finite number of cuts
		cannot solve any reasonable cake cutting challenge 
		in time that is bounded by any function of $\epsilon$ and $n$ on all possibly-unbounded-density 
		valuations.
	\end{lemma}
	\begin{proof}
		Assume by way of contradiction that the number of queries 
		made by the protocol is finite (for some fixed $n$ and $\epsilon$), i.e. the ``computation tree" of the
		protocol has finite depth. Since every query in the tree can have only a finite number of answers,
		every node in the tree has a finite number of children, and so by Konig's lemma, the whole
		computation tree is finite.  Since the protocol always makes only a finite number of cuts, the
		total number of locations $x$ at which the protocol ever makes a cut is finite -- call this number $M$.
		
		Now if we consider the possible input valuation profiles where all players have their valuation 
		concentrated on the sub-interval $[x,x+\delta]$ then, since the challenge is 
		reasonable, the protocol is required to be able to output a cut in every such interval.  But when 
		$\delta<1/M$, more than $M$ such possible answers are required.
	\end{proof}

\subsection{Why Bound the Number of Pieces?}

It turns out that if a protocol is not bounded in the number of pieces that
it can produce, then even $\epsilon$-perfect cake-cutting becomes very easy.

\medskip

\begin{lemma} \label{lem:num_pieces_bound}
	There exists a randomized public-coin protocol that produces
	an $\epsilon$-perfect allocation among any constant $n$ number of players 
	using $C=O(\epsilon^{-2})$ pieces with high probability without any communication.  
\end{lemma}  
\begin{proof}
	The unit interval is pre-split
	to $O(\epsilon^{-2})$ intervals of equals lengths (the hidden constant in the big-O notation depends polynomially on $n$). 
	The public 
	random coins are used to allocate each interval to one of the 
	$n$ players at random.  With high probability each player
	will get a total weight of within a few standard deviations from $1/n$. The bounded density assumption implies that
	each sub-interval has $O(\epsilon^{-2})$ weight for each of the players which implies that the standard deviation here is $O(\epsilon)$.  Thus the allocation specified by the public random coins is $\epsilon$-perfect with high probability.
\end{proof}

While this protocol uses public coins, the standard simulation 
of \cite{Newman91}
implies a simultaneous private-coin protocol 
with $O(\log\epsilon^{-1})$ bits of communication.

\medskip

\begin{open} \label{open:det_perfect_sim}
	What is the deterministic simultaneous communication complexity of finding a perfect allocation
	with $C=poly(\epsilon^{-1})$ cuts?
\end{open}

\section{Basic Simulations}

We provide several simulations, which use the observation that valuations with bounded density can be approximated by simple piecewise-constant-density valuations.

\begin{definition} \label{def:simple_valuations}
	A valuation $v'$ is called $m$-simple if (a) for every integer $0 \le k < m$, we have that
	$v'([k/m, (k+1)/m])$ is itself an integer multiple of $1/m$ and (b) for every $0 < x < 1$ we have that 
	$v'([k/m,(k+x)/m]) = x \cdot v'([k/m, (k+1)/m])$.
\end{definition}

Initially each player will approximate its valuation function
by a simple one.

\begin{lemma} \label{simple}
	For any valuation $v$, there exists an $m$-simple valuation $v'$ such that for every sub-interval $[x,y]$ 
	we have that
	$|v([x,y])-v'([x,y])| \le \epsilon$, where $m=O(\epsilon)$.
\end{lemma}
\begin{proof}
	Let $m$ be the smallest integer greater than $2/\epsilon$ and define, for every integer $0 \le k \le m$, 
	that $v'([0,k/m]) = \lc v([0,k/m]) \rc_\epsilon$, where $\lc y \rc_\epsilon$ is the smallest 
	integer multiple of $\epsilon$ that is greater or equal to $y$. Then interpolate and
	for every $0 < x < 1$ set 
	$v'([k/m,(k+x)/m]) = x \cdot v'([k/m, (k+1)/m])$. The valuation $v'$ obtained has the required property.
\end{proof}

From this we can easily deduce the following general simulation, which in particular implies that bounded RW protocols can be simulated efficiently
in the communication model.

\begin{lemma} \label{lem:any_simultaneous}
	For any notion of fairness $\mathcal{F}$ that is guaranteed to exist and has at most a constant $C$ number of cuts, 
	an $\epsilon$-approximate version of this notion can be found by a simultaneous protocol that uses
	$O(\epsilon^{-1}\log\epsilon^{-1})$ bits of communication.
\end{lemma}
\begin{proof}
	Each player approximates its valuation by an $m$-simple one, as in Lemma \ref{simple}, for $m=O(C \cdot \epsilon^{-1})$, and sends the complete description
	of the simple approximation. The assured fair allocation for the approximating valuations is found, and by Lemma \ref{simple} (and the definition of abstract fairness) it is also $\epsilon$-fair
	for the original valuations.
\end{proof}

\subsection{Simulating the RW Model}

Cake cutting protocols usually operate in a query model attributed to Robertson and Webb \cite{RW98} (and explicitely stated by Woeginger and Sgall \cite{WS07}). In this model, the protocol communicates with the players using the following types of queries:
\begin{itemize}
	\item{}$\emph{\textbf{Cut}}_i(\alpha)$: Player $i$ cuts the cake at a point $y$ where $v_{i}(0,y) = \alpha$. The point $y$ becomes a {\em cut point}.
	\item{}$\emph{\textbf{Eval}}_i(y)$:  Player $i$ returns $v_{i}(0,y)$, where $y$ is a previously made cut point.
\end{itemize}
Note the questions and answers have infinite precision. 
Any RW protocol asks the players a sequence of cut and evaluate queries, at the end of which it outputs an allocation where the pieces are demarcated by cut points. 


We show a general simulation theorem for RW protocols as follows. First, we say that two valuations $v,v'$ are $\epsilon$-close if $|v([x,y]) - v'([x,y])| \leq \epsilon$ for each interval $[x,y]$.

\begin{lemma} \label{sim-rw}
	Let $\mathcal{M}$ be any protocol in the RW model which makes at most $r$ queries and produces an allocation of the cake with a number of cuts bounded by a constant $C$. Then there exists a communication protocol $\mathcal{M}_{\epsilon}$ with $r$ rounds of $O(\log \epsilon^{-1})$ bits of communication each, such that for every valuation tuple $v$, $\mathcal{M}_{\epsilon}(v) = \mathcal{M}(v')$, for a valuation tuple $v'$ that is $\epsilon/C$-close to $v$.
\end{lemma}
\begin{proof}
	Each player $i$ that was given a valuation $v_i$ will first approximate it by a $m$-simple 
	valuation $v'_i$ as shown in Lemma \ref{simple} with $m=O(C \cdot \epsilon^{-1})$, where $C$
	is the maximum number of intervals that the cake can be cut into by protocol $\mathcal{M}$. 
	The players will run the simulated protocol below on
	the $v'_i$'s and as Lemma \ref{simple} guarantees that for every sub-interval 
	$|v_i([x,y])-v'_i([x,y])| \le \epsilon/C$, and each player gets at most $C$ connected intervals, the lemma
	will follow.
	
	To run the RW protocol $\mathcal{M}$ on the approximate valuations $v'$, note that the answer 
	to any RW query on
	these types of ``simple" valuations requires only $O(\log m)=O(\log \epsilon ^{-1})$ bits of communication.
	Specifically, the exact real-valued answer to a query ``$v'_i([0,x])=?$'' with $k/m \le x < (k+1)/m$
	can be given by specifying
	the values $v'_i([0,k/m])$ and $v'_i([0,(k+1)/m])$ which are integer multiples of $1/m$ and thus require
	$O(\log m)$ bits of communication.  Similarly the answer to a query 
	``$v'_i([0,?])=x$'' can either be given by an integer $k$ such that $v'_i([0,k/m])=x$ if such
	an integer exists, or otherwise by specifying the single integer $k$ such that 
	$v'_i([0,k/m]) < x < v'_i([0,(k+1)/m])$ as well as these two values 
	themselves, still $O(\log m)$
	bits of communication.
	Thus the players can
	simulate the RW protocol using $r$ rounds of $O(\log \epsilon^{-1})$ bits per round, which implies a communication protocol $\mathcal{M}_{\epsilon}$ with the required properties.
\end{proof}

By the definition of the closeness of $v'$ to $v$ and relying on the bounded number of cuts, we immediately get:

\begin{theorem} \label{sim-rw-continuous}
	Let $\mathcal{M}$ be any protocol in the RW model that makes at most $r$ queries and outputs $\mathcal{F}$-fair allocations. Then there exists a communication protocol $\mathcal{M}_{\epsilon}$ with $r$ rounds of $O(\log \epsilon^{-1})$ bits of communication each, such that $\mathcal{M}_{\epsilon}$ outputs ($\epsilon$-$\mathcal{F}$)-fair allocations.
\end{theorem}

Applying the known bounded protocols in the RW model for computing proportional and envy-free allocations \cite{EP84,AM16} we get:

\begin{corollary}\label{easy}
	For every $\epsilon > 0$,
	$\epsilon$-proportional cake cutting and $\epsilon$-envy-free cake cutting among any fixed number $n$ of players can be done with $O(1)$ rounds of $O(\log{\epsilon^{-1}})$ bits of communication.
\end{corollary}

We note that our simulation theorem falls short of approximately simulating arbitrary RW protocols for 
arbitrary problems since information can be lost when approximating by simple valuations.  This is demonstrated by the following example
for $n=2$ players.

\begin{quote}
	\begin{itemize}
		\item Ask Alice to cut the cake in half: $x = Cut_1(0.5)$. 
		
		\item If $x$ is irrational, Alice takes the whole cake. 
		
		\item Else, Bob gets its favorite piece in $\{[0,x], [x,1]\}$ and Alice the remainder.
	\end{itemize}
\end{quote}

Let $\epsilon > 0$. Consider hungry valuations $\vec{v} = (v_{Alice}, v_{Bob})$ such that Alice's midpoint is $x \in \Re \setminus \mathbb{Q}$. When running the protocol on valuations $\vec{v}$, Alice takes the whole cake.
However, when the players approximate their valuations by some simple valuations $\vec{v}'$, with $m = O(\epsilon^{-1})$ intervals, Alice's midpoint of the cake with respect to her new valuation $v_{Alice}'$ is a number $x' \in \mathbb{Q}$, since $1/2 \in \mathbb{Q}$, $m \in \mathbb{N}$, and her value for each interval $[i/m,(i+1)/m]$ is an integer multiple of $1/m$. Thus the only kind of execution observed on simple valuations is the computation of envy-free allocation between the two players.
\footnote{This type of discrepancy can also be observed for fair protocols, such as the RW implementation of the Dubins-Spanier protocol, where each player submits their $1/n$ mark (i.e. the point $x_i$ such that $v_i([0,x_i]) = 1/n$), then the player with the leftmost mark receives the piece $[0,x_i]$, and the protocol is repeated with the remaining $n-1$ players on the leftover cake. If a player ``trembles" when answering the cut queries, it can get pieces worth very different values when running the protocol on simple valuations compared to the original valuations. However, in this case, the allocation remains approximately proportional in both scenarios.}


\subsection{Simulating Moving Knife Procedures} \label{moving}

A family of protocols that has remained undefined until recently is that of moving knife procedures, 
which involve multiple knives sliding over the cake until some stopping condition is met (such as Austin's moving knife procedure for finding a perfect allocation between two players \cite{RW98}). In a companion paper \cite{BN17} we  formalized moving knife protocols that can be approximately
simulated in the RW model. Here we provide a more general definition of a class of protocols that can be approximately simulated with little
communication.  

An example of a moving knife protocol is Austin's procedure \cite{RW98}, which computes a perfect allocation between two players with two cuts (which is the minimum number of cuts required for perfect allocations) and operates as follows:
\medskip
\begin{quote}
	\textbf{Austin's procedure}:
	\emph{A referee slowly moves a knife from left to right across the cake. At any
		point, a player can call stop.
		When a player called, a second knife is placed at the left edge of the cake.
		The player that shouted stop -- say 1 -- then moves both knives parallel
		to each other.
		While the two knives are moving, player 2 can call stop at any time.
		After 2 called stop, a randomly selected player gets the portion between player
		1's knives, while the other one gets the two outside pieces.
	}
\end{quote}

The correctness is obtained by using the Intermediate Value Theorem. If at the beginning of the operation player 2 values the piece between the knives at less than 1/2, then his value for the remainder of the cake is more than 1/2. Then by the intermediate value theorem there exists a position of the sword such that player 2's value for the piece between the knives is exactly 1/2.
In general, moving knife protocols can have multiple knifes and multiple players that move them across the cake until some stopping condition is met.

The high level description of our formal model has an ordered collection of ``knives'' each of them controlled by a single player, who can place his
knife on the cake based on the locations of the previous knives in the collection as well as his own valuation.  Each player moves his knife continuously
as a function of ``the time'' and the locations of the previous knives, until some condition happens.  
The basic idea of the simulation is that (a) Given the position of previous
knives (and the current time), each player can announce the (approximate) position of his knife using a small ($O(\log \epsilon^{-1})$) 
amount of communication (b)
to find the (approximate) time where the condition happens, the players can engage in a simple binary search which takes $O(\log \epsilon^{-1})$
rounds (in each of which each player had to announce his location).

The details of the definition balance between ensuring that everything is smooth and robust enough so that
this general idea can work and allowing enough flexibility to capture a wide family of reasonable protocols including all known ones.
In particular, our definition allows multiple ``moving knife steps'' with ``simple'' and ``robust'' decisions between these steps, where all
known moving knife protocols require only a constant number of steps.  The exact definition appears in Appendix \ref{sec:appendix_mk}, and it
allows the following simulation statement.



\begin{theorem} \label{thm:mk_protocol_sim}
	Let $\mathcal{M}$ be any moving knife protocol that runs in constantly many steps and outputs $\mathcal{F}$-fair partitions.
	
	Then for every $\epsilon > 0$, there is a communication protocol $\mathcal{M}_{\epsilon}$ that uses 
	$O(\log \epsilon^{-1})$ rounds each with $O(\log \epsilon^{-1})$ bits of communication and
	outputs $\epsilon$-$\mathcal{F}$-fair partitions. 
\end{theorem}


This implies that procedures such as the Austin \cite{Austin82} procedure, the Barbanel-Brams, Stromquist, and Webb moving knife protocols \cite{BB04,RW98}, and the moving knife step for computing equitable allocations among any number of players \cite{BN17,SS17} can be simulated efficiently in the communication model.

\begin{corollary} \label{cor:mk_sim_all}
	Austin's procedure for computing a perfect allocation between two players, the Barbanel-Brams, Stromquist, and Webb procedures for computing an envy-free allocation with connected pieces among $n=3$ players, and the moving knife step for computing an equitable allocation for any fixed number $n$ of players can be simulated using $O(\log{\epsilon^{-1}})$ rounds
	each with $O(\log{\epsilon^{-1}})$ bits of communication for all $\epsilon > 0$.
\end{corollary}

Notice that this simulation requires $O(\log \epsilon^{-1})$ rounds of communication, as
opposed to the simulation of the RW model that only requires a constant number of rounds. Our lower bounds 
(Corollary \ref{med-lower}) will imply that this is unavoidable.

\section{Proportional Cake Cutting} 

We turn to the problem of computing allocations for specific fairness notions and start with the most basic, known as proportionality. 

\subsection{A General Logarithmic Lower Bound} 

A logarithmic (in $\epsilon^{-1}$) lower bound is trivial not just for communication protocols but for any model with a bounded
number of answers to every query. This lower bound applies to essentially any cake cutting challenge, and we will prove it for the simplest one of proportional cake cutting between
any fixed number $n \ge 2$ of players.  

\medskip

\begin{theorem} \label{thm:proportional_lower_bound}
For every fixed $n \ge 2$, any protocol whose queries all have a bounded number of possible answers and always makes at most $C$ cuts requires $\Omega(\log \epsilon^{-1})$ queries for $\epsilon$-proportional cake cutting between $n$ of players.
\end{theorem}
\begin{proof}
	We will show that the computation tree of the protocol requires 
	$\epsilon^{-\Omega(1)}$ leaves. Let $t$ be the total number of leaves of
	the protocol and let $X=\{x_1 < x_2 < \cdots < x_k\}$ be the set of all cuts made at all leaves
	of the protocol (thus $k \le C \cdot t$).  
	
	We will now choose a profile of valuations on which it is impossible
	that the protocol achieves $\epsilon$-proportionality.  In this chosen
	profile, all players will have the exact same valuation 
	$v_1 = \cdots = v_n = v$.
	We will choose $v$ in a way that puts a weight that is
	an integer multiple of $1/R$, for some fixed integer $R$, in every sub-interval
	$[x_i,x_{i+1}]$.  That would imply that every player gets a total
	value that is an integer multiple of $1/R$. 
	If we choose $R$ to be relatively prime to $n$ then certainly
	for every integer $m$ we would have
	$|m/R - 1/n| > 1/(nR)$.  If $1/(nR) > \epsilon$ 
	then at least one player would have
	to get a value that is less than $1/n$ by more than $\epsilon$.
	
	We will choose $R \approx 1/(10n\epsilon)$ which is relatively prime 
	to $n$
	which ensures that $1/(nR) > \epsilon$  and now need
	to construct the required valuation $v$.  Our goal is to assign
	to every sub-interval $[x_i,x_{i+1}]$ a weight that is an integer
	multiple of $1/R$ that is close to its length $x_{i+1}-x_i$.    This
	is done by first rounding $x_{i+1}-x_i$ down to an integer multiple
	of $1/R$ and setting it as (our initial version of) $v([x_i,x_{i+1}])$.  This is not quite
	a valuation since we may be missing a total weight that is at 
	most $k/R$ (a rounding loss of at most $1/R$ for each of the $k$ sub-intervals).  Now, if $k < \sqrt{R} = O(\epsilon^{-1/2})$ then 
	the longest of the $k$ sub-intervals is longer than 
	$k/R$ and we can put all the missing weight on it, increasing 
	the density on it from at most $1$ to at most $2$ satisfying the
	condition of bounded density. 
\end{proof}

\subsection{A Simultaneous Upper Bound}

The number of rounds in the general RW simulation is the same as the number of queries in the original RW protocol. In some cases one can get very 
simple communication protocols, such as the following, which is derived from a simultaneous protocol from \cite{BBKP14}.

\medskip

\begin{theorem} \label{thm:proportional_upper_bound}
For every fixed $n$, there exists a protocol for finding a connected $\epsilon$-proportional allocation between $n$ players, where the players 
	simultaneously, in a single round, send $O(\log \epsilon^{-1})$
	bits of communication and the allocation is determined by the sent
	messages.
\end{theorem}
\begin{proof}
	The players first approximate their valuations by $m$-simple
	ones (for $m=O(\epsilon^{-1})$) and then each
	player $i$ sends the $n-1$ locations that split its simple valuation
	into $n$ pieces each of weight exactly $1/n$. (As in the
	proof of Lemma \ref{sim-rw}, the simplicity
	of the approximating valuation allows representing each of these
	locations using $O(\log \epsilon^{-1})$ bits.)
	
	We can now use this information to construct a proportional
	allocation among the approximating valuations: denote
	$x_i^j$ to be the $j$ location sent by player $i$, and for
	notational convenience define $x_i^0=0$ and $x_i^n=1$ for 
	all players.
	Thus we have that $v_i([x_{j-1},x_j])=1/n$ for all $i=1 \ldots n$
	and $j=1 \ldots n$ (here $v_i$ is the approximate simple valuation,
	not the original one.)
	
	Let $i_1 = \arg \min_i x_i^1$ and let $x_*^1=x_{i_1}^1$
	and allocate $[0,x_*^1]$ to
	$i_1$. Let $i_2 = \arg \min_{i \ne i_1} x_i^2$ and $x_*^2=x_{i_2}^2$,
	and allocate $[x_*^1,x_*^2]$ to $i_2$.  Similarly, at each stage
	we allocate the next piece the the player who wants the smallest 
	amount. The correctness of this allocation follows from the
	fact that at each step $k$ of the allocation, the allocated
	piece of the cake was to
	the left of $x_i^k$ for all unallocated players $i$.
\end{proof}

\begin{open} \label{open:simultaneous}
What is the simultaneous communication complexity of
	$\epsilon$-envy-free cake cutting between $n=3$ players (into
	an arbitrary finite number of pieces)?
\end{open}

\section{Equitable, Perfect, and Envy-Free Cake-Cutting} \label{sec:eqpfef}

Finally we study fairness notions \footnote{We focus on allocations that use a minimum number of cuts.} for which moving knife protocols exist: equitable and perfect allocations for $n=2$ players, and envy-free for $n=3$ players.

\subsection{The Crossing Problem}

We will relate these fair division problems through a discrete search problem called crossing that we now introduce. 

\begin{definition} \label{def:crossing_problem}
	The {\em general} $crossing(m)$ problem is the following: Alice gets a sequence of integer numbers $x_0,x_1, \ldots ,x_m$ with $0 \le x_i \le m$ for all $i$ and similarly Bob gets 
	a sequence of integer numbers $y_0,y_1, \ldots ,y_m$ with $0 \le y_i \le m$ for all $i$, and such that $x_0 \le y_0$ and $x_m \ge y_m$. Their goal is to find an index $1 \le i \le m$ 
	such that either both $x_{i-1} \le y_{i-1}$ and $x_i \ge y_i$ or that both
	$x_{i-1} \ge y_{i-1}$ and $x_i \le y_i$.\footnote{An even more difficult variant of this problem
		would insist that the output is the former rather than the latter.   This variant does not
		seem to correspond directly to any fair division problem so we will not name it.  It is worth
		noting that the protocol in the proof of theorem \ref{ub} solves this harder variant as well.}
	 	
	In the monotone variant, $mon-crossing(m)$, Alice and Bob are also assured that $0=x_0 \le x_1 \le x_2 \le \cdots \le x_m=m$ and $m=y_0 \ge y_1 \ge \cdots \ge y_m=0$ and again need to find
	an index $1 \le i \le m$ 
	such that $x_{i-1} \le y_{i-1}$ and $x_i \ge y_i$.
\end{definition}

The $\epsilon$-equitable cake cutting problem can be seen as equivalent to monotone crossing as follows.

\begin{lemma} \label{lem:equitable_to_crossing}
	For $n=2$ players, the $\epsilon$-equitable cake-cutting problem into $C=2$ connected pieces can be reduced to the $mon-crossing(m)$
	problem for $m=O(\epsilon^{-1})$ with an additional single round where each player send $O(\log m)$ bits to the other one.
\end{lemma}
\begin{proof}
	The players will start by rounding their valuations to $m$-simple
	ones as allowed by Lemma \ref{simple} (with $m=O(\epsilon^{-1})$),
	and will achieve an equitable
	allocation on these simple valuations which by Lemma \ref{simple} also approximately
	solves the problem on the original valuations.  So for the rest of this proof let us assume that each $v$ is already $m$-simple.
	
	Now, for a fixed player, if we define $x_i=v([0,i/m])\cdot m$, we certainly get that 
	$0=x_0 \le x_1 \le x_2 \le \ldots \le x_{m-1} \le x_m = m$.  
	Let us denote, just for Bob, $y_i = m-x_i$ so
	$m = y_0 \ge y_1 \ge y_2 \ge \cdots \ge y_{m-1} \ge y_m = 0$.  (And for Alice we keep the $x$ notation.)
	
	A solution for the $mon-crossing(m)$ problem on these vectors $x_0 \ldots x_m$ and $y_0 \ldots y_m$ will give an index $i$ such that $x_{i-1} \le y_{i-1}$ and $x_i \ge y_i$ and thus
	$v_{Alice}([0, (i-1)/m])) + v_{Bob}([0, (i-1)/m]) \le 1$ while $v_{Alice}([0, i/m]) + v_{Bob}([0, i/m]) \ge 1$ which implies that for some intermediate value $(i-1)/m \le x* \le i/m$ we have 
	exactly $v_{Alice}([0, x*]) + v_{Bob}([0, x*]) = 1$ and moreover since simple valuations, by definition, are linear between $(i-1)/m$ and $i/m$, the exact value of $x*$ can be computed
	from the values $x_{i-1}, x_i, y_{i-1}, y_i$ which the players can send to each other in a single round.  
	Now splitting the cake at location $x*$ is an equitable allocation since $v_{Bob}([x*,1]) = 1 - v_{Bob}([0,x*]) = v_{Alice}([0,x*])$.
\end{proof}

\begin{lemma} \label{lem:crossing_to_equitable}
	The $mon-crossing(m)$ problem can be reduced (without any additional communication) into the $\epsilon$-equitable cake cutting problem among $n=2$ players into $C=2$ connected pieces, with $\epsilon^{-1}=O(m^2)$.
\end{lemma}
\begin{proof}
	Given her input $0=x_0 \le x_1 \le x_2 \le \cdot \le x_{m-1} \le x_m = m$, Alice will construct the following valuation $v_{Alice}$: 
	\begin{itemize}
		\item $v_{Alice}([0,1/3])=v_{Alice}([2/3,1])= (1-1/m)/2$.  (I.e. only $1/m$ weight is left for the sub-interval $[1/3,2/3]$.)
		\item Constant density within $[0,1/3]$ and $[2/3,1]$.  I.e. the density in these two sub-intervals is $3 \cdot (1-1/m)/2)$.
		\item For every integer $0 \le i \le m$, let $v_{Alice}([1/3,1/3 + i/(3m)]) = x_i/m^2$.  
		\item Constant density within each subinterval $[1/3+i/(3m), 1/3+(i+1)/(3m)]$. I.e. the density is $3 \cdot (x_{i+1}-x_i)/ m \le 3$.  
	\end{itemize}
	
	Bob will similarly create his own valuation $v_{Bob}$ using the values $x_i^{Bob}=m-y_i$.  Notice that for each $i$ we have that 
	$v_{Alice}([1/3,1/3 + i/(3m)]) + v_{Bob}([1/3,1/3 + i/(3m)]) = (x_i + m - y_i)/m^2$, and so 
	$v_{Alice}([0,1/3 + i/(3m)]) + v_{Bob}(0,1/3 + i/(3m)]) = 1-1/m + (x_i + m - y_i)/m^2$.  In particular, due to the integrality of $x_i$ and $y_i$, if $x_i>y_i$ then 
	$v_{Alice}([0,1/3 + i/(3m)]) + v_{Bob}(0,1/3 + i/(3m)]) > 1+1/m^2$ while if $x_i < y_i$ then
	$v_{Alice}([0,1/3 + i/(3m)]) + v_{Bob}(0,1/3 + i/(3m)]) < 1-1/m^2$.  
	
	Now consider a solution $(i-1)/m \le x* \le i/m$ to the $\epsilon$-equitable cake cutting problem 
	where $v_{Alice}([0,x*])-v_{Bob}([x*,1]) | \le \epsilon$, equivalently, $1-\epsilon \le v_{Alice}([0,x*])+v_{Bob}([0,x*]) \le 1+\epsilon$.  
	If $\epsilon<1/m^2$ then this cannot be the case if either $x_i \ge x_{i-1} > y_{i-1} > y_i$ or $x_{i-1} \le x_i < y_i \le y_{i-1}$.  It thus
	follows that $x_{i-1} \le y_{i-1}$ and $x_i \ge y_i$ as needed. 
\end{proof}

Next we show that $\epsilon$-perfect cake cutting is equivalent to the general crossing problem.
The omitted proofs of this section can be found in Appendix \ref{app:crossing_problem}.

\begin{lemma} \label{lem:perfect_to_crossing}
	For $n=2$ players, the $\epsilon$-perfect cake-cutting problem into $C=3$ intervals can be reduced to the general $crossing(m)$
	problem for $m=O(\epsilon^{-1})$ with a single additional round of $O(\log \epsilon^{-1})$ bits
	of communication.
\end{lemma}

\begin{lemma} \label{lem:crossing_to_perfect}
	The general $crossing(m)$ problem can be reduced (without any additional communication) into the $\epsilon$-perfect cake cutting problem 
	among $n=2$ players into $C=3$ intervals, with $\epsilon^{-1}=O(m^2)$.
\end{lemma}

Finally, we show $\epsilon$-envy-free cake cutting for $n=3$ players can be reduced to monotone crossing.

\begin{lemma} \label{lem:envy-free_to_crossing}
For $n=3$ players, the $\epsilon$-envy-free cake-cutting problem with connected pieces can be reduced to the $mon-crossing(m)$
problem for $m=O(\epsilon^{-1})$ with an additional constant number of rounds of $O(\log{\epsilon^{-1}})$ bits of communication.
\end{lemma}

\subsection{The Communication Complexity of the Crossing Problem}

We study the communication complexity of the crossing problem and start with an upper bound. 

\begin{theorem} \label{ub}
	The deterministic communication complexity of $crossing(m)$ is $O(\log^2 m)$. The randomized communication complexity of $crossing(m)$ is $O(\log m \log\log m)$.
\end{theorem}
\begin{proof}
	The deterministic upper bound follows from our general simulation result of moving knife procedures, but here is a simple protocol: the players perform a binary search over the indices. At each point during the search they have a sub-interval of indices where $x_a \le y_a$ and $x_b \ge y_b$, 
	looking at the middle point $c$ between $a$ and $b$, once they determine whether $x_c \ge y_c$ or
	$x_c \le y_c$ they know which sub-interval to continue the search in. Since $x_c$ and $y_c$ are
	integers in the range $0 \ldots m$, comparing $x_c$ and $y_c$ requires $O(\log m)$ bits of communication.
	The binary search itself requires $O(\log m)$ such comparisons.
	
	For the randomized protocol, we invoke the results of \cite{nis93} stating that the {\em randomized}
	communication complexity of determining which of two $k$-bit integers (here we have $k=\log m$) is larger is $O(\log k)$. Moreover, with that complexity the probability of error can be 
	made reduced to be $1/k^c$ for any desired constant $c$. Thus each of the comparisons in our protocol can be done, with probability of error, say, $1/(\log^2 m)$ by a randomized protocol with $O(\log\log m)$ bits of communication. Since we make at most $\log m$ such comparisons, the total
	probability of error still stays low, and the total number of bits of communication is as required.
\end{proof}

For the monotone variant we give a more efficient protocol that uses only $O(\log{m})$ communication. Towards this end, we consider two parameters for the problem: $m$, the number of integers held by the players and $k$, the upper bound on their values.

\begin{definition}
	The $mon-crossing(m,k)$ problem is the following: Alice gets a sequence of
	integer numbers $0 = x_0 \le x_1\le x_2 \le \cdots \le x_m = k$ and  Bob gets a sequence of
	integer numbers $k = y_0 \ge y_1 \ge \cdots \ge y_m = 0$.
	Their goal is to find an index $1 \le i \le m$ such that $x_{i-1} \le y_{i-1}$
	and $x_i \ge y_i$. 
\end{definition}

\begin{theorem}
	The (deterministic) communication complexity of the $mon-crossing(m,k)$ 
	problem is $O(\log m + \log k)$.
	In particular when $m=k$ (what we called the $mon-crossing(m)$ problem), the
	communication complexity is $O(\log m)$.
\end{theorem}
\begin{proof}
	In the first step Alice and Bob each send a single bit describing whether $x_{m/2} < k/2?$ and
	$y_{m/2} < k/2?$.  Now there are four possible answers:
	\begin{itemize}
		\item
		If $x_{m/2} < k/2$ and $y_{m/2} < k/2$ then define 
		$x'_i = min(k/2,x_i)$ and $y'_i = min(k/2, y_i)$.  Notice that 
		$0 = x'_0 \le x'_1\le x'_2 \le \cdots \le x'_m = k/2$ and  
		$k/2 = y'_0 \ge y'_1 \ge \cdots \ge y'_m = 0$, so we solve the problem recursively on
		the $x'_i$'s and $y'_i$'s and find an index $i$ such that $x'_{i-1} \le y'_{i-1}$
		and $x'_i \ge y'_i$.  
		
		Now if $i \le m/2$ then $x_{i-1} \le x_i \le x_{m/2} < k/2$ and so 
		$x'_{i-1}=x_{i-1}$ and $x'_i = x_i$ and since $y'_i \le x'_i < k/2$ also $y'_i = y_i$,
		while for $i-1$ we only know that $y_{i-1} \ge y'_{i-1}$.  So we get that
		$x_{i-1}=x'_{i-1} \le y'_{i-1} \le y_{i-1}$ and $x_i=x'_i \ge y'_i=y_i$ as needed.
		
		On the other hand  if $i > m/2$ i.e. $i-1 \ge m/2$ then $y_i \le y_{i-1} \le y_{m/2} < k/2$ and so 
		$y'_{i-1}=y_{i-1}$ and $y'_i = y_i$ and since $x_{i-1} \le y_{i-1} < k/2$ also $x'_{i-1} = x_{i-1}$,
		while for $i$ we only know that $x_{i} \ge x'_{i}$.  So we get that
		$x_{i-1}=x'_{i-1} \le y'_{i-1} = y_{i-1}$ and $x_i \ge x'_i \ge y'_i=y_i$ as needed.
		
		\item
		If $x_{m/2} \ge k/2$ and $y_{m/2} \ge k/2$ then define 
		$x'_i = max(k/2,x_i)-k/2$ and $y'_i = max(k/2, y_i)-k/2$.  Notice that 
		$0 = x'_0 \le x'_1\le x'_2 \le \cdots \le x'_m = k/2$ and  
		$k/2 = y'_0 \ge y'_1 \ge \cdots \ge y'_m = 0$, so we solve the problem recursively on
		the $x'_i$'s and $y'_i$'s and find an index $i$ such that $x'_{i-1} \le y'_{i-1}$
		and $x'_i \ge y'_i$.
		
		Now if $i > m/2$ i.e. $i-1 \ge m/2$ then $x_{i} \ge x_{i-1} \ge x_{m/2} \ge k/2$ and so 
		$x'_{i-1}=x_{i-1}-k/2$ and $x'_i = x_i-k/2$ and since $y'_{i-1} \ge  0$ also $y'_{i-1} = y_{i-1}-k/2$,
		while for $i$ we only know that $y'_{i} \ge y_i-k/2$.  So we get that
		$x_{i-1}=x'_{i-1}+k/2 \le y'_{i-1}+k/2 = y_{i-1}$ and $x_i=x'_i+k/2 \ge y'_i+k/2 \ge y_i$ as needed.
		
		On the other hand  if $i \le m/2$ then $y_{i-1} \ge y_{i} \ge y_{m/2} \ge k/2$ and so 
		$y'_{i-1}=y_{i-1}-k/2$ and $y'_i = y_i-k/2$ and since $x'_{i} \ge 0$ also $x'_{i} = x_{i}-k/2$,
		while for $i-1$ we only know that $x'_{i-1} \ge x'_{i-1}-k/2$.  So we get that
		$x_{i-1} \le x'_{i-1}+k/2 \le y'_{i-1}+k/2 = y_{i-1}$ and $x_i=x'_i+k/2 \ge y'_i+k/2 =y_i$ as needed.
		
		\item
		If $x_{m/2} \ge k/2$ and $y_{m/2} < k/2$ then change the values $y_{m/2}=0$
		and $x_{m/2}=k$
		and solve the problem recursively on 
		$0=x_0 \le x_1 \le \cdots \le x_{m/2-1} \le x_{m/2} = k$ and 
		$k=y_0 \ge \cdots \ge y_{m/2-1} \ge y_{m/2} = 0$ 
		and find an index $i$ such that $x_{i-1} \le y_{i-1}$
		and $x_i \ge y_i$ (with the new values of $x_{m/2}$ and $y_{m/2}$). 
		If $i < m/2$ then clearly we have the same inequalities with the old values as well.
		If $i=m/2$ then since with the original values of $x_{m/2}$ and $y_{m/2}$ we have that
		$x_{m/2} \ge k/2 > y_{m/2}$ then $i$ is a solution for the original problem as well.
		
		\item
		If $x_{m/2} < k/2$ and $y_{m/2} \ge k/2$ then 
		change the values $y_{m/2}=k$
		and $x_{m/2}=0$
		and solve the problem recursively on 
		$0=x_{m/2} \le x_{m/2+1} \le \cdots \le x_{m-1} \le x_{m} = k$ and 
		$k=y_{m/2} \ge y_{m/2+1} \ge \cdots \ge y_{m-1} \ge y_{m} = 0$ 
		and find an index $i$ such that $x_{i-1} \le y_{i-1}$
		and $x_i \ge y_i$ (with the new values of $x_{m/2}$ and $y_{m/2}$). 
		If $i-1 > m/2$ then clearly we have the same inequalities with the old values as well.
		If $i-1=m/2$ then since with the original values of $x_{m/2}$ and $y_{m/2}$ we have that
		$x_{m/2} < k/2 \le y_{m/2}$ then $i$ is a solution for the original problem as well.
	\end{itemize}
	
	Now notice that in the first two cases the recursive problem had a value of $k$ that was half
	of the original problem, while in the last two cases, the recursive problem had a value of $m$
	that is half of the original value. Thus we get the upper bound on communication
	of $CC(m,k) \le 2 + max(CC(m/2,k), CC(m,k/2))$ and the lemma follows.	
\end{proof}

Applying the equivalence of the two variants of the crossing problem to our fair division problems we get:

\begin{corollary} \label{cor:bounds_fairness}
	The deterministic communication complexity of finding an $\epsilon$-perfect allocation with two cuts between $n=2$ players is $O(\log^2 \epsilon^{-1})$. The randomized communication complexity of this problem is $O(\log \epsilon^{-1} \log\log \epsilon^{-1})$.
\end{corollary}

\begin{corollary}\label{cor:bounds_fairness_mon}
	The deterministic communication complexity of finding a connected $\epsilon$-equitable allocation 
	between $n = 2$ players and a connected $\epsilon$-envy-free allocation 
	among $n = 3$ players is $O(\log \epsilon^{-1})$.
\end{corollary}

At this point one may wonder whether the communication complexity of the general $crossing(m)$ problem is also
in fact logarithmic, maybe even for deterministic protocols. We leave this as an open question.

The
main apparent difference though between the two fair division problems considered here and the ``easy'' fair division problems mentioned in Corollary \ref{easy} is the fact that it seems that the binary
search process here (as well as in the general simulation of moving knife procedures in Section \ref{moving}) cannot be done in a constant number of rounds.  Our main technical lower bound proves
that this is indeed the case.

\begin{theorem} \label{thm:bound_lower}
	Any (deterministic or randomized) protocol that solves $mon-crossing(m)$ using rounds of 
	communication of
	$t$-bits each (with $\log m \le t \le m$) requires $\Omega(\log m/ \log t)$ rounds of communication.
\end{theorem}

\begin{corollary} \label{med-lower}
	Any (deterministic or randomized) protocol for finding an $\epsilon$-perfect allocation with 2-cuts between $n=2$ players or for finding an $\epsilon$-equitable allocation between $n=2$ players using rounds of  communication of
	$polylog(\epsilon^{-1})$-bits each requires $\Omega(\log \epsilon^{-1}/ \log \log \epsilon^{-1})$ rounds of communication.
\end{corollary}

Our lower bound will follow the general ``round elimination technique'' formalized by the
the ``round elimination lemma'' of \cite{mnsw95}, which was somewhat strengthened in \cite{ss08}.

\begin{definition} \label{def:jointproblem}
Given a communication problem $f(x,y)$ and a parameter $k$ we define a new communication problem,
denoted by $P^k(f)$ as
follows: Alice gets $k$ strings $x^1, \ldots, x^k$ and Bob gets a string
$y$ and an integer $1 \le z \le k$. Their aim is to compute $f(x^z , y)$.
\end{definition}

Suppose a protocol for this new problem is given, where
Alice goes first, sending Bob $t$ bits, where $t<<k$. 
Intuitively, it would seem that since Alice does not
know $z$, the first round of communication cannot be productive.
It turns out that this is correct.   

\begin{lemma} \label{lem:round_elimination_lemma}
	({\bf Round Elimination Lemma}, \cite{mnsw95,ss08}): 
	Suppose there is a $[t,r]$ protocol with Alice going first for solving
	$P^k(f)$ with probability of error $\delta$ then there is a 
	$[t,r-1]$-protocol (with Bob going first)
	for solving $f$ with probability of error $\delta+\sqrt{t/k}$.
\end{lemma}

In order to use this lemma repeatedly towards contradiction, each time removing another round of communication, we will require a lemma showing how to reduce the $P^k(mon-crossing(m))$ problem into a single instance of $mon-crossing(mk)$.

\begin{lemma} \label{lem:pk_to_crossing}
The $P^k(mon-crossing(m))$ problem can be reduced, without any communication, to a single $mon-crossing(km)$ problem.
\end{lemma}
\begin{proof}
	Alice holds $k$ vectors of length $m+1$ each $(x_0^1 \ldots x_m^1) \ldots (x_0^k \ldots x_m^k)$ and will produce from them a single vector as follows: for every $i=0 \ldots m$ and every $j=1 \ldots k$ we set
	$X_{(j-1) \cdot m + i} = x_i^j + (j-1)m$. Notice that indices that are integer multiples of $m$ are defined twice, once as $X_{(j-1)m} = x_0^j + (j-1)m$ and also as $X_{(j-2)m}=x_m^{j-1} + (j-2)m$, but since $x^j_0=0$ and $x^{j-1}_m=m$ these values are the same.  Also notice that the $X$'s are monotone increasing.
	
	Bob that get $y_0 \ldots y_m$ and an index $z$ will produce the longer vector as follows:
	$Y_0=Y_1=\cdots=Y_{(z-1)m}=mk$, $Y_{zm}=Y_{zm+1}=\cdots=Y_{km}=0$ and for all $i=0 \ldots m$, 
	$Y_{(z-1)m+i}= y_i + (z-1)m$.  Notice that $Y$ is monotone decreasing.
	
	Now let us look at a solution for the $mon-crossing(km)$ problem on $X$ and $Y$, which is
	some location $w$ where $Y_w \ge X_w$ but $Y_{w+1} \le X_{w+1}$. First we
	must have $w+1 > (z-1)m$ (since $Y_{(z-1)m} = km > (z-1)m = X_{(z-1)m}$) and also
	$w < zm$ (since $Y_{zm}=0 < zm =X_{zm}$). Thus $w=(z-1)m+i$ for some $0 \le i < m$ so
	$X_w=X_{(z-1)m+i} = x_i^z + (z-1)m$ and $Y_w = Y_{(z-1)m+i} = y_i + (z-1)m$ and similarly
	$X_{w+1} = x_{i+1}^z + (z-1)m$ and $Y_{w+1} = y_{i+1} + (z-1)m$.  It follows that
	$y_i \ge x_i^z$ but $y_{i+1} \le x_{i+1}^z$ as needed for the solution of our 
	$P^k(mon-crossing(m))$ problem.
\end{proof}

We can now complete the proof of the main lower bound.

\begin{lemma} \label{lem:biground}
Any (deterministic or randomized) protocol that solves $mon-crossing(m)$ in $r$ rounds with probability of error $1/3-1/(10r)$ must send at least $t=\Omega(m^{1/r}/r^4)$ bits of communication in at least one of these rounds.
\end{lemma}
\begin{proof}
	Let us start by proving the base case, $r=1$. 
	By the lemma above, a 1-round protocol for $mon-crossing(m)$ (with, say, probability of error $1/3$)
	implies a 1-round protocol for $P^{m/3}(mon-crossing(3))$ (with the same error) which, if
	$t < m/300$ implies
	a 0-round protocol for $mon-crossing(3)$ with an error larger by an additive $1/10$, but
	a 0-round protocol uses no communication so it cannot solve the $mon-crossing(m)$ problem even
	for a constant $m=3$ (say) and error $1/3+1/10$.

	Now for the induction step: assume by way of contradiction
	a $r$-round protocol for $mon-crossing(m)$ that uses $t < m^{1/r}/(100r^4)$ bits per round and has 
	probability of error at most $\epsilon \le 1/3-1/(10r)$.  Take 
	$k=m^{1/r}$ so by the two lemmas above a $r$-round protocol for $mon-crossing(m)$ implies
	a similar $r$ round protocol for $P^k(m^{(r-1)/r})$ which implies
	a $r-1$ round protocol for $m'= m^{(r-1)/r}$ that uses $t$ bits of communication per round 
	with probability of error 
	$\epsilon + 1/10r^2$.
	But notice that $(m')^{1/(r-1)} = m^{1/r}$ so
	so we get a $(r-1)$-round protocol for $mon-crossing(m')$ with 
	$t < (m')^{1/(r-1)}/(100r^4) \le (m')^{1/(r-1)}/(100(r-1)^4)$
	and error of at most $(1/3-1/(10r)) + 1/1(10r^2) \le 1/3 - 1/(10(r-1))$ which is impossible by the induction hypothesis.
\end{proof}

We note that the exact form of the error probability is only for the convenience of the induction. Clearly the same bound holds for any constant error.

\section{Discussion}
A remaining interesting direction is to understand the joint dependence on $n$ and $\epsilon$. We summarize the upper bounds for deterministic protocols (both from previous literature and here) in the next table.

\begin{table}[h!]
	\label{tab:sum}
	\begin{center}
		\begin{tabular}{|c | c|c|c|c ||} 
			
			\hline \hline
			Fairness notion & Number of players & Upper bounds for deterministic protocols \\
			\hline \hline
			$\epsilon$-proportional (connected) & $n \geq 2$ & $O\left(n \log{n} \cdot \log{\epsilon^{-1}}\right)$ \cite{EP84} \\ \hline 
			$\epsilon$-envy-free (any number of cuts) & $n \geq 2$ &  $O\left(n^{n^{n^{n^{n^{n}}}}} \cdot \log{\epsilon^{-1}}\right)$ \cite{AM16} \\ \hline
			\multirow{3}{*}{\centering{$\epsilon$-envy-free (connected)}}& $n=2$ &  $O(\log{\epsilon^{-1}})$ \cite{RW98} \\
			& $n = 3$ & $O(\log^2{\epsilon^{-1}})$ \cite{BN17} \\ 
			& $n \geq 4$ & $O(n \cdot \epsilon^{-1}\log{\epsilon^{-1}})$ (\textcolor{RubineRed}{$*$}) \\
			\hline
			\multirow{2}{*}{\centering{$\epsilon$-equitable (connected)}}& $n=2$ &  $O(\log^2{\epsilon^{-1}})$ \cite{CP12} \\
			& $n \geq 3$ &  $O\left(\min\{n, \epsilon^{-1}\} \cdot \log^2{\epsilon^{-1}}\right)$ \cite{CP12,PW17} \\  \hline 
			\multirow{2}{*}{\centering{$\epsilon$-perfect (minimum  cuts)}}& $n=2$ & $O(\log^2{\epsilon^{-1}})$\
			\cite{BN17}  \\ 
			& $n \geq 3$ &  $O\left(n^3\cdot \epsilon^{-1} \log{\epsilon^{-1}}\right)$ (\textcolor{RubineRed}{$*$}) \\ \hline \hline
		\end{tabular}
		\label{tab:summary}
	\end{center}
	\caption{Communication complexity upper bounds for deterministic protocols. The number of players is arbitrary. The bounds marked with blue references are obtained by simulating in the communication model protocols from the RW model. The bounds marked with (\textcolor{RubineRed}{$*$}) are implied by the general simulation in this paper (Proposition \ref{lem:sim_all}) of fairness notions that are guaranteed to exist with a bounded number of cuts. }
\end{table}

\newpage

\nocite{*}
\addcontentsline{toc}{section}{\protect\numberline{}References}%

\bibliographystyle{alpha}

\bibliography{cake}

\newpage

\appendix

\section{Equitable, Perfect, and Envy-Free Cake-Cutting}

\subsection{The Crossing Problem} \label{app:crossing_problem}

In this section we include the omitted proofs of the reductions between the crossing problem and fair division problems.

\medskip

\noindent \textbf{Lemma \ref{lem:perfect_to_crossing}.}
\emph{For $n=2$ players, the $\epsilon$-perfect cake-cutting problem into $C=3$ intervals can be reduced to the general $crossing(m)$ problem for $m=O(\epsilon^{-1})$ with a single additional round of $O(\log \epsilon^{-1})$ bits of communication.}

\begin{proof}
	Fix $m=5D/\epsilon$, where $D$ is the upper bound on the density of the players' valuations.
	For ease of notation let each of the two players slightly perturb its valuation by at most 
	$1/m$ so that $v([0,k/m])=v([k/m,1])=1/2$ for some integer $k$ (that may be different for each of
	the two players).    We will find a $4\epsilon/5$-perfect cutting for the perturbed valuations
	which implies an $\epsilon$-perfect one for the original valuations.
	
	In a single round of communication the players exchange their values of $k$, and without loss
	of generality assume that Alice's $k$ is no larger than Bob's $k$, and from now on we will let
	$k$ be Alice's value.  Thus $v_{Alice}([0,k/m])=v_{Alice}[k/m,1])=1/2$ while 
	$v_{Bob}([0,k/m]) \le 1/2$ and $v_{Bob}[k/m,1]) \ge 1/2$. 
	
	Alice chooses $x_0 \le x_1 \le \cdots \le x_k$ so that
	$x_i \in \{0,1, \ldots, m\}$ is the integer that minimizes
	$|v_{Alice}([i/m,x_i/m])-1/2|$.  Note that since the density is at most $D$, 
	the maximum weight of an interval of length $1/m$ is $D/m \le \epsilon/5$ so we sure
	get $|v_{Alice}([i/m,x_i/m])-1/2| \le \epsilon/5$.  Bob chooses in a similar fashion
	$y_0 \le y_1 \le \cdots \le y_k$.  Note that since $v_{Bob}([k/m,1]) \ge 1/2$, then
	the $y_i$'s are still well defined and are within the range $[k,m]$ 
	even though we are using Alice's value of $k$.   
	Thus $x_0=k$ and $x_k=m$ while
	$y_0 \ge k$ and $y_k \le m$ so if we pad $x_k=x_{k+1}=\cdots=x_m$ and $y_k=y_{k+1}=\cdots=y_m$
	then
	we get vectors that satisfy the input conditions to the 
	general $crossing(m)$ problem.\footnote{Note that even though both vectors are monotone, we still
		don't get an instance of the monotone crossing problem since here both vectors are increasing
		rather than one of them decreasing.}
	
	Now assume that we get a solution to the 
	general $crossing(m)$ problem which wlog we assume is $x_{i-1} \le y_{i_1}$ and $x_i \ge y_i$ (the other case is similar).  We claim that cutting the cake
	at $i/m$ and at $x_i/m$ and handing one of the players the middle part is an approximately
	perfect cake cutting.   By our construction $v_{Alice}([i/m,x_i/m])-1/2| \le \epsilon/5$.
	We can calculate $v_{Bob}([i/m,x_i/m]) = v_{Bob}([i/m,y_i/m])-v_{Bob}([x_i/m,y_i/m])$.
	For the first term we have by our construction $|v_{Bob}([i/m,y_i/m])-1/2| \le \epsilon/5$.
	Since $y_{i-1} \le x_{i-1} \le x_i$, the second term can be bounded by 
	$v_{Bob}([x_i/m,y_i/m]) \le v_{Bob}([y_{i-1}/m,y_i/m])$ which in turn can be bounded
	by $v_{Bob}([y_{i-1}/m,y_i/m]) = v_{Bob}([i/m, y_i/m])- v_{Bob}([(i-1)/m, y_{i-1}/m]) -
	v_{Bob}([(i-1)/m, i/m])$.  The first two terms are each within $\epsilon/5$ from $1/2$
	and the second term is at most $D/m \le \epsilon/5$ thus in total we have bounded
	$|v_{Bob}([i/m,x_i/m])-1/2| \le 4\epsilon/5$ as needed.
\end{proof}

\medskip

\noindent \textbf{Lemma \ref{lem:crossing_to_perfect}.}
\emph{The general $crossing(m)$ problem can be reduced (without any additional communication) into the $\epsilon$-perfect cake cutting problem among $n=2$ players into $C=3$ intervals, with $\epsilon^{-1}=O(m^2)$.}
\begin{proof}
	The two players are given their inputs $x_0 \ldots x_m$ and $y_0 \ldots y_m$ to the $crossing(m)$ problem and need to construct valuations $v_A$ and $v_B$ from them. We
	first assume without loss of generality (by padding) that $x_0=0$, $x_m=m$. 
	
	For simplicity of presentation each of the two players will slightly modify their input by adding the quantity $2m^2 + m + 2mi$ to each $x_i$ and $y_i$.  
	Notice that the range of the
	$x_i$'s and $y_i$'s is now $\{H,2H\}$, where $H=2m^2+m$ and each vector is monotone increasing: 
	$m \le x_{i+1}-x_i \le 3m$.  The required
	answer to the $crossing(m)$ problem with these new values is still the same as for the original problem since each pair of $x_i$ and $y_i$ were increased by the same amount.  From now on we will 
	use the notation $x_i$ and $y_i$ to be the new modified values rather than the original ones.  (So now we have $x_0=H$ and $x_m=2H$ while $y_0 \ge H$ and $y_m \le 2H$.)
	
	Alice will define her valuation $v_{Alice}$ as follows: 
	\begin{itemize}
		\item $v_{Alice}([0,x_0/(2H)])=1/2$ with uniform density inside this interval.
		\item For every $i=1 \ldots m$, $v_{Alice}([x_{i-1}/(2H),x_i/(2H)]) = 1/(2m)$, with uniform density inside this interval.  (Note that the density is bounded from below and from above by a constant.)
	\end{itemize}
	
	Bob will create his valuation $x_{Bob}$ similarly but with the $y$ values.
	This implies that for every $i=0 \ldots m$ we have $v_{Alice}([i/(2m),x_i/(2H)])=1/2$, and 
	$v_{Bob}([i/(2m),y_i/(2H)])=1/2$. Now consider some solution to the $\epsilon$-perfect allocation problem.  I.e.
	some interval such that 
	$|v_{Alice}([a,b])-1/2| \le \epsilon$ and $|v_{Bob}([a,b])-1/2| \le \epsilon$.
	Let us assume that $(i-1)/(2m) \le a \le i/(2m)$ for $1 \le i \le m$ and present
	$a = \alpha (1-i)/(2m) + (1-\alpha)i/(2m)$ 
	for some $0 \le \alpha \le 1$.\footnote{We are ignoring here the theoretical possibility that
		$a$ is larger than $1/2$ which, since $v_{Alice}([1/2,1])=1/2$ and $v_{Alice}$ has density bounded
		from below by a constant, can only be true when $a \le 1/2 + O(\epsilon)$, in which case the
		inequality holds for $i=m$ to within an additive $O(\epsilon)$ term, which suffices for the rest of the proof.}    So we must also have
	$b = \alpha x_{i-1}/(2H) + (1-\alpha) x_i/(2H) + O(\epsilon)$ (since $v_{Alice}([i/(2m),x_i/(2H)])=1/2$ and $v_{Alice}$ has uniform density within
	each of the sub-intervals $[(i-1)/(2m),i/(2m)]$ and $[x_{i-1}/(2H),x_i/(2H)]$ and since 
	the density is
	bounded by a constant 
	from below and from above.) For the same reason, using $v_{Bob}$ 
	we must have $b = \alpha y_{i-1}/(2H) + (1-\alpha) y_i/(2H) + O(\epsilon)$.
	Now notice that if both $x_{i-1}>y_{i-1}$ and $x_i>y_i$ then since these are integers then 
	$\alpha x_{i-1}/(2H) + (1-\alpha) x_i/(2H)$ must be 
	larger than $\alpha y_{i-1}/(2H) + (1-\alpha) y_i/(2H)$ by at least $1/(2H)$ (and the opposite must)
	be true in the dual case where $x_{i-1}<y_{i-1}$ and $x_i<y_i$).  This is impossible if the 
	$O(\epsilon)$ term is smaller than $1/(2H)$ in which case it follows that 
	either $x_{i-1} \ge y_{i-1}$ and $x_i \le y_i$ or vice versa, as required.
\end{proof}

We also show that $mon-crossing$ can be used to solve the envy-free cake cutting problem with connected pieces for three players.

\medskip

\noindent \textbf{Lemma \ref{lem:envy-free_to_crossing}.}
\emph{For $n=3$ players, the $\epsilon$-envy-free cake-cutting problem with connected pieces can be reduced to the $mon-crossing(m)$
	problem for $m=O(\epsilon^{-1})$ with an additional constant number of rounds of $O(\log{\epsilon^{-1}})$ bits of communication.}
\begin{proof}		
	Given a cake cutting problem among Alice, Bob, and Carol, we will construct a monotone crossing instance whose solution will reveal an approximate solution to the envy-free allocation among the three players. 
	Again each player starts by perturbing its valuation to make it $m$-simple, for $m = 10 D/\epsilon$. We will find an $\epsilon/2$-envy-free allocation for the $m$-simple valuations, which will automatically give an $\epsilon$-envy-free allocation for the original valuations. 
	
	Let the players exchange their midpoints of the cake ($m_A,m_B,m_C$), the $1/3$ marks (i.e. the points $\ell_A, \ell_B, \ell_C$ for which $v_{Alice}([0,\ell_A])$ $=$ $v_{Bob}([0,\ell_B])$ $=$ $v_{Carol}([0,\ell_C]) = 1/3$), and the $2/3$ marks ($r_A, r_B, r_C$). This can be done in one round with $O(\log{\epsilon^{-1}})$ bits of communication.
	
	W.l.o.g., assume that $\ell_A \leq \ell_B \leq \ell_C$. If there is an $\epsilon/2$-envy-free allocation with cuts $\ell_A$ and $r_A$, return it. These preferences can be determined with $O(1)$ communication. 
	Otherwise, both Bob and Carol prefer one of the pieces $[\ell_A, r_A]$ or $[r_A, 1]$ to the outside options by more than $\epsilon/2$. We have two cases: \\
	
	$\bullet$ \emph{Case 1:} Bob and Carol prefer the piece $[\ell_A, r_A]$. 
	\medskip
	
	If $m_A \leq m_B$, Alice and Bob calculate $j = \lfloor m \cdot \ell_A \rfloor$ and $k = \lfloor m \cdot m_A \rfloor$.
	If $m_A > m_B$, they set $j = \lfloor m \cdot \ell_A \rfloor$ and
	$k = \lfloor m \cdot m_B\rfloor$. We will show the details for $m_A \leq m_B$; the arguments for the other case are similar. Clearly $0 \leq j \leq k < m$.
	Alice defines her sequence $x$ starting with $x_{j} = \lfloor m \cdot r_A \rfloor$.
	The first term, $x_j$, is set to an integer for which the difference $|v_{Alice}([0, j/m]) - v_{Alice}([x_j/m,1])|$ is minimized. Each additional term $x_i$ is defined as an integer smaller than or equal to $x_{i-1}$ for which 
	$|v_{Alice}([0, i/m]) - v_{Alice}([x_i/m,1])|$ is minimal. 
	Since each interval $[\ell/m,(\ell+1)/m]$ is worth at most $\epsilon/10$ to either player, we have that $|v_{Alice}([0, i/m]) - v_{Alice}([x_i/m,1])| \leq \epsilon/10$. 
	
	Bob will use his valuation to define a weakly increasing sequence $y$ as follows. For each $i = j, \ldots, k$, set $y_i$ as the integer for which the piece $[i/m,y_i/m]$ is as close as possible from Bob's point of view to the largest outside option (i.e. $[0,i/m]$ or $[y_i/m,1]$), such that $y_i \geq y_{i-1}$. In case of ties, the first term, $y_{j}$, is chosen as small as possible.
	The sequence $y$ can always be set to be weakly increasing for the following reason. On the range of indices $i$ for which $i/m \in [\ell_A, \ell_B]$, the terms $y_i$ are increasing since the piece $[0,i/m]$ is worth less than $1/3$ to Bob, the piece $[i/m,1]$ is shrinking as $i$ increases, and so the (approximate) tie must be created between $[i/m,y_i/m]$ and $[y_i/m,1]$. On the range of indices $i$ for which $i/m \geq \ell_B$, the $y_i$ terms continue to grow as $i$ increases since the largest of $[i/m,y_i/m], [y_i/m,1]$ must now match the (increasing) value of the piece $[0,i/m]$, which has become the largest (approximately).
	We obtain that for all $i = j \ldots k$,
	$\left|v_{Bob}([i/m,y_i/m]) - \max\left\{v_{Bob}([0,i/m]), v_{Bob}([y_i/m,1])\right\} \right| \leq \epsilon/10.$
	
	Finally, Alice sets the prefix of her sequence to m and the suffix to zero, while Bob does so reversely (i.e. $x_1 = \ldots = x_{j-1} = m = y_{k+1} = \ldots = y_m$ and $y_1 = \ldots = y_j = 0 = x_{k+1} = \ldots x_m$). 
	
	We show the crossing terms for $x$ and $y$ must be in the range $j \ldots k$.
	By construction $v_{Alice}([x_{j}/m,r_A]) \leq \epsilon/10$ and $v_{Alice}([x_{k}/m, m_A]) \leq \epsilon/10$. We have $y_{j} \leq x_{j}$, since $x_j$ is such that for Bob $[j/m,x_{j}/m]$ is larger than the outside options by more than $\epsilon/10$, while $[j/m,y_{j}/m]$ is approximately tied with one of $[0,j/m]$ or $[y_j/m]$ for Bob.
	On the other hand, $y_{k} \geq x_{k}$ since $v_{Bob}([k/m,x_k/m]) \leq \epsilon/10 < 1/3 - \epsilon/10 \leq v_{Bob}([k/m,y_k/m])$. 
	
	Suppose we are given a solution $x_{i-1} \geq y_{i-1}$ and $x_i \leq y_i$ for the $mon-crossing(m)$ problem with sequences $x$ and $y$. The existence of a crossing at index $i$ implies $S = [x_{i}, x_{i-1}] \cap [y_{i-1},y_i] \neq \emptyset$. Let $z$ be the  maximum element of $S$. Then $x_i \leq z \leq x_{i-1}$ and $y_{i-1} \leq z \leq y_i$.
	
	We claim that cutting at $i/m$ and $z/m$ is an $\epsilon/2$-envy-free solution with respect to the simple valuations. By the bound on the densities, the choice of $m$ and of the sequences $x$ and $y$, we have that $[(i-1)/m,i/m]$ is worth at most $\epsilon/10$ for both Alice and Bob. Since Alice views the pieces $[0,(i-1)/m]$ and $[x_{i-1}/m,1]$, and $[0,i/m]$ and $[x_{i}/m,1]$, respectively, as tied within $\epsilon/10$, we get $v_{Alice}([x_{i-1}/m,x_i/m]) \leq 3\epsilon/10$.
	For Bob we similarly have $v_{Bob}([y_{i-1}/m,y_i/m]) \leq 3\epsilon/10$. 
	From $-\epsilon/10 \leq v_{Alice}([x_i/m,1]) - v_{Alice}([0,i/m]) \leq \epsilon/10$ and
	$v_{Alice}([z/m,1])$ $=$ $v_{Alice}([x_i/m,1]) - v_{Alice}([x_i/m,z/m])$, we obtain 
	\begin{eqnarray*}
		- 4\epsilon/10 & \leq & -\epsilon/10 - v_{Alice}([x_i/m,z/m]) \leq v_{Alice}([z/m,1]) - v_{Alice}([0,i/m]) \\
		& = & v_{Alice}([x_i/m,1]) - v_{Alice}([x_i/m,z/m]) - v_{Alice}([0,i/m])
		\leq \epsilon/10 - v_{Alice}([x_i/m,z/m]) \\
		& < & 4\epsilon/10
	\end{eqnarray*}
	The pieces $[0,i/m]$ and $[z/m,1]$ are tied for Alice within $4\epsilon/10$.
	Moreover, since Alice values $[i/m,x_i/m]$ for at most $1/3-\epsilon/10$, the piece $[i/m,z/m]$ is worth at most $1/3 + 2\epsilon/10$ to her, which implies that $\min\{v_{Alice}([0,i/m]), v_{Alice}([z/m,1])\} \geq v_{Alice}([i/m,z/m]) - 5\epsilon/10$. Thus giving Alice one of $[0,i/m]$ or $[z/m,1]$ would lead to her having envy of at most $\epsilon/2$ for any other player. For Bob, if the tie is between $[i/m,y_i/m]$ and $[y_i/m,1]$, then $|v_{Bob}([i/m,y_i/m]) - v_{Bob}([y_i/m,1])| \leq \epsilon/10$. From $v_{Bob}([z/m,y_i/m]) \leq 3\epsilon/10$, we get that $|v_{Bob}([i/m,z/m])$ $-$ $v_{Bob}([z/m,1])|$ $\leq$ $4\epsilon/10$ and $|v_{Bob}([i/m,z/m])- v_{Bob}([0,i/m])| \leq 4 \epsilon/10$. The error term when Bob's tie is between $[0,i/m]$ and $[i/m,y_i/m]$ is similar, which implies Bob has two largest pieces tied within $\epsilon/2$, one of which is $[i/m,z/m]$.
	
	These cuts are revelead to Carol, who chooses her favorite piece among $[0,i/m]$, $[i/m,z/m]$, and $[z/m,1]$. If Carol chooses $[0,i/m]$, then Alice takes $[z/m,1]$ and Bob $[i/m,z/m]$. If Carol chooses $[i/m,z/m]$, Bob takes his favorite from $[0,i/m]$ and $[z/m,1]$, while Alice gets the remaining piece. If Carol chooses $[z/m,1]$, Alice picks $[0,i/m]$ and Bob $[i/m,z/m]$. \\
	
	$\bullet$ \noindent \emph{Case 2:} Bob and Carol prefer the piece $[r_A,1]$ to $[0, \ell_A]$ and $[\ell_A, r_A]$ by more than $\epsilon/2$. 
	\medskip
	
	Then $v_{Bob}([r_B,1]) > 1/3$, so $r_A < r_B$.
	Alice and Bob compute $j = \lfloor r_A \cdot m \rfloor$
	and $k = \lceil r_B \cdot m \rceil$.
	For each $i = j \ldots k$, Alice sets $x_i$ as the integer minimizing for her the difference in values between $[0,x_i/m]$ and $[x_i/m,i/m]$. Then $|v_{Alice}([0,x_i/m]) - v_{Alice}([x_i/m,i/m])| \leq \epsilon/10$ and the sequence $x$ is weakly increasing since the interval $[0,i/m]$ grows with $i$.
	For each $i = j \ldots k$, Bob sets $y_i$ as the integer that equalizes as much as possible for him the piece $[i/m,1]$ with $[0,y_i/m]$. Then $|v_{Bob}([0,y_i/m]) - v_{Bob}([i/m,1])| \leq \epsilon/10$ and the sequence $y$ is weakly decreasing.
	
	To see why the sequences must cross in the range $j \ldots k$, note that $y_j \geq x_j$ since Bob values $[0,x_j/m]$ less than $1/3 - \epsilon/10$ (by $\ell_A \leq \ell_B$) while $[0,y_j/m]$ is worth to him at least $1/3 - \epsilon/10$.
	We must also show $y_k \leq x_k$. Suppose the opposite holds, $y_k > x_k$. Then Bob views the pieces $[0,y_k/m]$, $[y_k/m,i/m]$, and $[k/m,1]$ as equal within $\epsilon/10$. Moreover, Carol cannot prefer $[0,y_k/m]$ to the other pieces by more than $\epsilon/10$, since $\ell_C > \ell_B$. Then we can immediately obtain an $\epsilon/2$-envy-free allocation by letting Carol take her favorite piece between $[y_k/m,k/m]$ and $[k/m,1]$, giving Alice $[0,y_k/m]$ (which is worth to her at least as much as $[0,x_k/m]$ and $[k/m,1]$), and Bob the remaining piece. This allocation can be found with $O(\log{\epsilon^{-1}})$ communication without even the players constructing the sequences. Thus the more difficult case is when $y_k \leq x_k$. Let Alice and Bob pad their sequences so that $x_1 = \ldots x_{j-1} = 0 = y_{k+1} = \ldots = y_m$ and $y_1 = \ldots y_{j-1} = m = x_{k+1} = \ldots x_m$.
	
	Suppose we have a solution to the $mon-crossing(m)$ problem with $x_{i-1} \leq y_{i-1}$ and $x_i \geq y_i$. Then $S = [x_{i-1}, x_i] \cap [y_{i}, y_{i-1}] \neq \emptyset$. Let $z$ be the maximum element of $S$.
	We will obtain an $\epsilon/2$-envy-free allocation by cutting at $z/m$ and $i/m$. As in Case 1, Alice and Bob view the pieces $[0,z/m]$, $[z/m,i/m]$ and $[0,z/m]$, $[i/m,1]$, respectively, as tied to be the largest within $\epsilon/2$. Carol is asked for her favorite piece. If Carol selects $[0,z/m]$, Alice takes $[z/m,i/m]$ and Bob $[i/m,1]$. If Carol selects $[z/m,i/m]$, Alice takes $[0,z/m]$ and Bob $[i/m,1]$. Otherwise, Carol takes $[i/m,1]$, Alice $[z/m,i/m]$ and Bob $[0,z/m]$.
	
	In both cases we obtained an $\epsilon/2$-envy-free allocation for the $m$-simple valuations, which is an $\epsilon$-envy-free solution with respect to the real valuations as required.
\end{proof}

\section{Moving Knife Protocols} \label{sec:appendix_mk}

A moving knife \emph{step} is a continuous operation that involves several devices moving continuously across the cake until some stopping condition is met.


\begin{definition} {\bf (A Moving Knife Step)}  \label{def:knife}
	There are a constant number $K$ of {\em Devices} some of which have a position on the cake
	and are called 
	{\em Knives} and others can have arbitrary real values and may be called {\em Triggers}.  
	The devices are numbered $1 \ldots K$ and 
	each device $j$ is controlled by a player $i_j$ and has a real value $x_j$ that changes continuously as {\em time} proceeds from $\alpha$ to $\omega$, where $0 \leq \alpha \leq \omega \leq 1$.
	Thus the value of each knife $j$ (i.e. its position on the cake) is given by some function 
	$x_j(t) \in [0,1]$, while the value of each trigger $j$
	is given by a function $x_j(t) \in \Re$. 
	
	\medskip
	
	The requirement is that the first device is a knife whose position is equal to the time, while the value of any additional device $j$ is $x_j(t) = F_j\left(t, x_1(t), \ldots, x_{j-1}(t)\right)$. The function $F_j$ may depend in an arbitrary way on the valuation $v_{i_j}$ of the player $i_j$ controlling the device and on any outcomes of previous steps of the protocol, but its dependence on the time $t$ and the values of previous devices $x_1(t) \ldots x_{j-1}(t)$ is Lipschitz continuous for every hungry valuation $v_{i_j}$.
	\medskip
	
	An outcome for the moving knife step is the index of a trigger $j$ with different signs at time $\alpha$ and $\omega$ (i.e. where $x_j(\alpha)\cdot x_j(\omega) \le 0$), together with a time $t$ such that $x_j(t)=0$, as well as the values of all devices $x_{j'}(t)$ at that time. (If the value of the trigger $x_j$ happens to be monotone then the time $t$ is unique, but in general there may be different such $t$ and any one of them may be given.)
\end{definition}

We may not always be able to find an exact outcome of a moving knife step with bounded communication, as the time at which a trigger fires can be irrational for example, so we will concentrate on obtaining approximate outcomes.

\begin{definition}[$\epsilon$-outcome of a moving knife step; see also \cite{BN17}]
	An $\epsilon$-outcome of a moving knife step running from time $\alpha$ to $\omega$ with devices $1 \ldots K$ is
	the index of a trigger $i$ that switches signs from $\alpha$ to $\omega$ (i.e. with $x_i(\alpha)\cdot x_i(\omega) \le 0$), together with a time $t$ and approximate values $\tilde{x}_1(t) \ldots \tilde{x}_K(t)$ of all the devices at this time, such that 
	\begin{itemize}
		\item $x_i(t) \in [-\epsilon,\epsilon]$, and 
		\item $|\tilde{x}_j(t) - x_j(t)| \leq \epsilon$, for all $j = 1 \ldots K$.
	\end{itemize}
\end{definition}

Moving knife steps can be simulated approximately in the communication model \footnote{We note that this simulation may find a time that is in fact far from the time where the trigger becomes exactly zero. This is unlike the RW model, where we can maintain the values of the devices with infinite precision, case in which for any point in time $t$ the values $x_1(t) \ldots x_K(t)$ are known exactly. Nevertheless, the moving knife protocols from the literature work for both notions of approximation.
} as follows.

\begin{theorem}
	Let $\mathcal{M}$ be a moving knife step.
	Then for all $\epsilon > 0$ and hungry valuations $v_1 \ldots v_n$, for every trigger of $\mathcal{M}$ that switches signs from the beginning to the end of time that $\mathcal{M}$ runs, there is a communication protocol $\mathcal{M}_{\epsilon}$ that finds an $\epsilon$-outcome associated with that trigger using
	$O\left( \log \epsilon^{-1}\right)$ rounds each with 
	$O\left(\log \epsilon^{-1}\right)$ bits of communication.
\end{theorem}
\begin{proof}
	Let the moving knife step run from time $\alpha$ to $\omega$ with devices $1 \ldots K$.
	Since the dependence functions of the devices are $\alpha$-Lipschitz, for some constant $\alpha$, the number of devices is constant, it follows that there exists a constant 
	$\zeta > \max\{1, \alpha\}$ such that for any device $i = 1 \ldots K$ and times $s < t$, where $s,t \in [\alpha, \omega]$, 
	\begin{equation} \label{ineq:inductive}
	|x_i(t) - x_i(s)| \leq  \zeta \cdot  |t-s|.
	\end{equation}
	
	Let $M = K+1$ be the maximum number of arguments of any function $F_i$.
	First, given $\epsilon > 0$ and time $t$, we note there exist approximate values $\tilde{x}_1(t) \ldots \tilde{x}_K(t)$ with $O(\log \epsilon^{-1})$ bits of all the devices at this time, such that $|\tilde{x}_i(t) - x_i(t)| \leq \epsilon/ \left((4 \cdot \zeta)^{K-i}\sqrt{M^{K-i}} \right)$ for each device $i$. For the first device this holds since the device is a knife with position equal to the time $t$, and so there is 
	$\tilde{x}_1(t) \in [0,1]$ with $O(\log \epsilon^{-1})$ bits such that 
	$|\tilde{x}_1(t) - t| \leq \epsilon/ \left((4 \cdot \zeta)^{K-1}\sqrt{M^{K-1}} \right)$.
	Suppose the condition holds for devices $1 \ldots s-1$. 
	
	\medskip
	
	We have the inequalities:
	\begin{eqnarray*}
		& & \left|F_{s}\left(t,x_1(t),\ldots, x_{s-1}(t)\right) - 
		F_s\left(t,\tilde{x}_1(t), \ldots, \tilde{x}_{s-1}(t)\right)
		\right| \\
		& \leq &
		\zeta \cdot ||\left(t,x_1(t),\ldots, x_{s-1}(t)\right) -
		\left(t,\tilde{x}_1(t), \ldots, \tilde{x}_{s-1}(t)
		\right)
		||_2 \\
		& \leq & \zeta \cdot  \sqrt{M} \cdot \frac{\epsilon}{2 \cdot 4^{K-s} \cdot \zeta^{K-s+1} \sqrt{M^{K-s+1}}}\\
		& = & \frac{\epsilon}{2 \cdot (4 \cdot \zeta)^{K-s} \sqrt{M^{K-s}}}
		%
	\end{eqnarray*}
	
	Then player $i_s$ controlling the device can report a value $\tilde{x}_s(t)$ with $O(\log \epsilon^{-1})$ bits such that $|\tilde{x}_s(t) - F_s\left(t,\tilde{x}_1(t), \ldots, \tilde{x}_{s-1}(t)
	\right)| \leq  \frac{\epsilon}{2 \cdot (4 \cdot \zeta)^{K-s} \sqrt{M^{K-s}}}$.
	The correct value of device $s$ at time $t$ is $x_s(t) = F_{s}\left(t,x_1(t),\ldots, x_{s-1}(t)\right)$; we obtain $|\tilde{x}_s(t) - x_s(t)| \leq  \frac{\epsilon}{(4 \cdot \zeta)^{K-s} \sqrt{M^{K-s}}}$ as required. Thus $|\tilde{x}_i(t) - x_i(t)| \leq \epsilon$ for each time $t$ and device $i$.
	
	\medskip
	
	To find a time when trigger $j$'s value is approximate zero, we will construct a protocol that maintains throughout the search two points in time at which the trigger has different signs. We are assured that $x_j(\alpha) \cdot x_j(\omega) \leq 0$, so initialize $s = \alpha$ and $t = \omega$. 
	By
	previous arguments, player $i_j$ can report $\tilde{x}_j(s)$ and $\tilde{x}_j(t)$ with $O(\log \epsilon^{-1})$ bits such that $|\tilde{x}_j(t') - x_j(t')| \leq \epsilon/2$, for $t' = s,t$.
	If $\tilde{x}_j(t') \in [-\epsilon/2, \epsilon/2]$ for some $t' \in \{s,t\}$, then $x_j(t') \in [-\epsilon,\epsilon]$, so the time $t'$ is a correct solution that can be output together with the values $\tilde{x}_1(t') \ldots \tilde{x}_K(t')$. Otherwise, 
	$\tilde{x}_j(t') \not \in [-\epsilon/2,\epsilon/2]$ for $t' = s,t$, which implies that $x_j(s),x_j(t) \not \in [-\epsilon,\epsilon]$. In this case the product $\tilde{x}_j(s) \cdot \tilde{x}_j(t)$ has the same sign as $x_j(s) \cdot x_j(t)$, and we can check the midpoint $(s+t)/2$. If $\tilde{x}_j(s) \in [-\epsilon/2,\epsilon/2]$, return $(s+t)/2$. Otherwise, continue the search in the interval between $[s, (s+t)/2]$ and $[(s+t)/2,t]$ where $\tilde{x}_j$ switches signs. If a solution has not been found after $2\zeta/\epsilon$ steps, stop and return the time $s$. Note that since $|\alpha - \omega| \leq 1$, after $2 \zeta/\epsilon$ steps we have $|s -t| \leq \epsilon/(2 \zeta)$, and so $|x_i(s) - x_i(t)| \leq \zeta \cdot |s - t| \leq \epsilon/2$. Moreover, we are guaranteed by the condition maintained throughout the search that $x_i$ switches signs in $[s,t]$, thus $\tilde{x}_i(s) \in [-\epsilon,\epsilon]$ and the time $s$ is a correct approximation obtained using $O(\log \epsilon^{-1})$ rounds each with $O(\log \epsilon^{-1})$ bits of communication.
\end{proof}

Note that even if the players have simple valuations, we may not be able to simulate exactly moving knife steps with very complex dependence functions. For example, a moving knife step $\mathcal{K}$ may be such that on some instances the only time when a trigger fires is an irrational point, case in which the simulation in the communication model is bound to return an approximate time. The moving knife step may be further embedded in a moving knife protocol with \emph{if} branches checking whether $\mathcal{K}$ was executed exactly or only approximately, and with radically different allocations in the two scenarios.
Thus we will require from moving knife protocols that they remain robust when the answers of the players are only approximately correct.

\begin{definition}[Moving Knife Protocol]
	A \emph{moving knife protocol} $\mathcal{M}$ consists of a finite number of steps $1 \ldots K$,
	where the outcome of the $k^{th}$ step is denoted by $R_k$ and represents either 
	\begin{itemize}
		\item any real number that a player can output as a function of its own valuation, or
		\item the outcome of a moving knife step
	\end{itemize}
	and in both cases may additionally depend on the history, i.e. on outcomes $R_1 \ldots R_{k-1}$ of the previous steps. At the end of each execution, $\mathcal{M}$ outputs an allocation of the cake
	that depends on the outcomes of the steps it executed.
	
	\smallskip
	We also require robustness: if $\mathcal{M}$ outputs $\mathcal{F}$-fair allocations, then for all $\epsilon > 0$, by iteratively replacing each outcome $R_k$ of step $k = 1 \ldots K$ of $\mathcal{M}$ with an $\epsilon$-outcome $\tilde{R}_k$,
	$\mathcal{M}$ still outputs $\epsilon$-$\mathcal{F}$-fair partitions and runs in $K$ steps.
\end{definition}

\medskip

It is easy to observe that the existing moving knife protocols (see, e.g., \cite{RW98}) are robust.

\medskip

\noindent \textbf{Theorem} \ref{thm:mk_protocol_sim} (restated). \emph{Let $\mathcal{M}$ be any moving knife protocol that runs in constantly many steps and outputs $\mathcal{F}$-fair partitions.
	Then for every $\epsilon > 0$, there is a communication protocol $\mathcal{M}_{\epsilon}$ that uses 
	$O(\log \epsilon^{-1})$ rounds each with $O(\log \epsilon^{-1})$ bits of communication and outputs $\epsilon$-$\mathcal{F}$-fair partitions. }
	\begin{proof}
	Let each player $i$ make its valuation hungry, if it is not already, by setting $v_i'(x) = (1 - \epsilon/2) \cdot v_i(x) + \epsilon/2$ for all $x \in [0,1]$. Otherwise, set $v_i' = v_i$. 
	
	Consider the following protocol $\mathcal{M}_{\epsilon}$, which follows the execution of $\mathcal{M}$ on valuations $v_i'$ as follows. For each step $k$ of $\mathcal{M}$, given approximate outcomes $\tilde{R}_1 \ldots \tilde{R}_{k-1}$ of the previous steps:
	\begin{itemize}
		\item if step $k$ requires some player $i$ to output a number, the player can communicate an $\frac{\epsilon}{2r}$-approximation, $\tilde{R}_k$, of the exact number it has computed with $O(\log \epsilon^{-1})$ bits of communication. 
		\item if step $k$ is a moving knife step, then $\mathcal{M}_{\epsilon}$ can set $\tilde{R}_k$ to an $\frac{\epsilon}{2r}$-outcome of the step with $O(\log \epsilon^{-1})$ rounds of $O(\log \epsilon^{-1})$ bits of communication.
	\end{itemize}
	By the robustness of the protocol, the allocation obtained is $\epsilon/2$-$\mathcal{F}$-fair with respect to valuations $v_i'$. Since $\mathcal{F}$ is an abstract fairness notion and the valuations $v_i'$ are $\epsilon/2$ close to $v_i$, this allocation is $\epsilon$-$\mathcal{F}$-fair with respect to the original valuations $v_i$ as required.
\end{proof}

\medskip
We obtain the following simulations for the known moving knife procedures.

\bigskip

\noindent \textbf{Corollary} \ref{cor:mk_sim_all}. \emph{
	Austin's procedure for computing a perfect allocation between two players, the Barbanel-Brams, Stromquist, and Webb procedures for computing an envy-free allocation with connected pieces among $n=3$ players, and the moving knife step for computing an equitable allocation for any fixed number $n$ of players \cite{BN17,SS17} can be simulated using $O(\log{\epsilon^{-1}})$ rounds
	each with $O(\log{\epsilon^{-1}})$ bits of communication for all $\epsilon > 0$.}

\end{document}